\documentclass[11pt,a4paper]{article}

\usepackage[truedimen,margin=30mm]{geometry}

\usepackage{mathrsfs}
\usepackage{amssymb}
\usepackage{amsmath}
\usepackage{ascmac}
\usepackage{amsthm}

\usepackage[pdftex]{graphicx}
\usepackage[pdftex]{color}

\usepackage{booktabs}
\usepackage{setspace}
\usepackage{natbib}
\usepackage{color}
\usepackage{xcolor}
\usepackage{url}

\usepackage{titlesec}
\titleformat*{\section}{\large\bfseries}
\titleformat*{\subsection}{\it}

\newtheorem{theorem}{Theorem}
\newtheorem{lemma}{Lemma}

\newtheorem{prp}{Proposition}

\newtheorem{condition}{Condition}

\def\ga{{\gamma}}

\def\si{{\sigma}}
\def\th{{\theta}}

\def\beh{{\widehat \beta}}

\def\thh{{\widehat \th}}

\def\psih{{\widehat \psi}}

\def\sh{\widehat{s}}
\def\Ah{\widehat{A}}

\def\st{\widetilde{s}}
\def\tht{\widetilde{\theta}}

\title{{\bf Adaptively Robust Small Area Estimation: Balancing Robustness and Efficiency of Empirical Bayes Confidence Intervals}}

\date{}
\author{}

\begin{document}

\maketitle
\doublespacing

\vspace{-1.5cm}
\begin{center}
{\large Daisuke Kurisu$^{1,4}$, Takuya Ishihara$^{2,4}$ and Shonosuke Sugasawa$^{3,4}$}

\medskip
\medskip
\noindent
$^1$Graduate School of International Social Sciences, Yokohama National University\\
$^2$Graduate School of Economics and Management, Tohoku University \\
$^3$Center for Spatial Information Science, The University of Tokyo\\
$^4$Nospare Inc. 
\end{center}

\vspace{2mm}
\begin{center}
{\bf \large Abstract}
\end{center}
Empirical Bayes small area estimation based on the well-known Fay-Herriot model may produce unreliable estimates when outlying areas exist. Existing robust methods against outliers or model misspecification are generally inefficient when the assumed distribution is plausible. This paper proposes a simple modification of the standard empirical Bayes methods with adaptively balancing robustness and efficiency. The proposed method employs $\gamma$-divergence instead of the marginal log-likelihood and optimizes a tuning parameter controlling robustness by pursuing the efficiency of empirical Bayes confidence intervals for areal parameters. We provide an asymptotic theory of the proposed method under both the correct specification of the assumed distribution and the existence of outlying areas. We investigate the numerical performance of the proposed method through simulations and an application to small area estimation of average crime numbers.


\newpage
\section{Introduction}
Direct survey estimators based only on area-specific samples are known to produce unacceptably large standard errors when the area-specific sample sizes are small. 
To improve the accuracy, empirical Bayes methods for small area estimation based on two-stage normal models are widely used to "borrow strength" from the information of other areas. 
For comprehensive overviews of small area estimation, see \cite{rao2015small}, and \cite{sugasawa2020small}.

A basic model for small area estimation is the two-stage normal hierarchical model known as the Fay-Herriot model \citep{fay1979estimates}, described as 
\begin{equation}\label{model}
y_i|\theta_i\sim N(\theta_i, D_i),  \ \ \ \ 
\theta_i\sim N(x_i^\top \beta, A), \ \ \ \ i=1,\ldots,m,
\end{equation}
where $y_i$ is the direct estimator of the small area parameter $\theta_i$ of interest, $D_i$ is known sampling variance, $x_i$ and $\beta$ are $p$-dimensional vectors of the auxiliary variables and regression coefficients, respectively, and $A$ is an unknown variance parameter. 
The empirical Bayes estimator of $\theta_i$ is useful when the areal parameter $\theta_i$ is well explained by the auxiliary information $x_i$.
However, the auxiliary information $x_i$ may not be useful for $\theta_i$ in all the areas, which could produce outlying areas for which the normality assumption for $\theta_i$ is not reasonable. 
In this case, the resulting estimator of $\theta_i$ can over-shrink the direct estimator $y_i$ in the outlying areas or be highly inefficient for non-outlying areas.

To circumvent the aforementioned problem, several robust methods have been proposed. 
In the context of small area estimation, \cite{sinha2009robust} and \cite{ghosh2008influence} proposed the use of Huber's $\psi$-function to modify the empirical Bayes estimator of $\theta_i$.
Furthermore, \cite{sugasawa2020robust} recently introduced a simple method by using density power divergence instead of the log-likelihood function of $y_i\sim N(x_i^\top\beta, A+D_i)$.
However, a serious bottleneck of the above approaches is that one needs to specify a tuning parameter adjusting the degree of robustness.
In this sense, the existing methods do not have adaptive nature; that is, there is the possibility that one might set an inappropriate degree of robustness, leading to inefficient estimation when the distributional assumption is plausible. 
Ideally, it is preferable to automatically optimize the tuning parameters depending on whether the distributional assumption is reasonable or not.
Moreover, the asymptotic theory of the existing robust methods is based on the assumption that the model is correctly specified, which is incompatible with the situation where the methods try to provide a remedy.

In this work, we introduce an adaptively robust approach to small area estimation. 
The key differences of the proposed method compared with the existing robust approaches are mainly three. 
First, we employ $\gamma$-divergence \citep{jones2001comparison}, which is known to have stronger robustness \citep[e.g.][]{fujisawa2008robust} than density power divergence \citep{basu1998robust} as used in \cite{sugasawa2020robust}.
Second, we derive robust empirical Bayes confidence intervals in addition to point estimators and propose a data-dependent way of selecting the tuning parameter of $\gamma$-divergence by pursuing the efficiency of the robust empirical Bayes confidence intervals. 
Third, we derive asymptotic properties of the proposed method under the existence of outlying areas. 
The second difference means that the proposed method is {\it adaptively robust}; that is, it reduces to the standard (non-robust) method when the distributional assumption is plausible. 
This property will be demonstrated in both theoretical and numerical ways in the subsequent sections. 
The third difference is novel in that this paper is the first one to study the theoretical properties of robust small area estimation methods under the existence of outlying areas.

We finally review related works beyond empirical Bayes small area estimation. 
As a general empirical Bayes approach, there are semiparametric approaches for prior distribution \citep{koenker2014convex, efron2016empirical}, but these methods do not directly handle auxiliary information as required in the small area estimation. 
More recently, \cite{armstrong2020robust} developed a semiparametric empirical Bayes confidence interval using only higher moment information, but their method can be conservative, especially when the assumed model is plausible.   
In small area estimation, \cite{datta1995robust} and \cite{tang2018modeling} proposed a fully Bayesian approach using heavy-tailed priors to take account of outlying areas, for which any asymptotic properties are not given, and Markov Chain Monte Carlo computation is required for estimation. 
Therefore, the proposed method seems advantageous in terms of simple computation, balancing robustness and efficiency, and valid theoretical properties under the existence of outlying areas.

\section{Adaptively robust small area estimation }
\label{sec:method}

\subsection{Fay-Herriot models and confidence intervals}
We consider the Fay-Herriot model (\ref{model}) and let $\psi = (\beta^{\top},A)^{\top}$ be the unknown parameter vector in the model (\ref{model}).
Then, the posterior (conditional) distribution of $\theta_i$ given $y_i$ and $\psi$ is $N(\tht_i, \st_i^2)$, where 
\begin{equation}\label{pos}
\tht_i(\psi)=y_i-\frac{D_i}{A+D_i}(y_i-x_i^\top \beta), \ \ \ \ \ \
\st_i^2(A)=\frac{AD_i}{A+D_i}.
\end{equation}
Since $y_i\sim N(x_i^\top\beta, A+D_i)$ under the model (\ref{model}), the unknown model parameter $\psi$ can be estimated as $\psih={\rm argmax}_{\psi}L(\psi)$, where $L(\psi)$ is the marginal likelihood expressed as 
\begin{equation}\label{ML}
L(\psi)=-\frac{1}{2}\sum_{i=1}^m\log\{2\pi(A+D_i)\} -\frac{1}{2}\sum_{i=1}^m\frac{(y_i-x_i^\top \beta)^2}{A+D_i}.
\end{equation}
Then, the empirical Bayes estimator of $\theta_i$ is obtained as $\thh_i=\tht_i(\psih)$.
Furthermore, for measuring uncertainty of $\thh_i$, we can construct the following empirical Bayes confidence intervals of $\theta_i$ with $1-\alpha$ nominal level \citep{cox1975prediction,morris1983parametric}:
\begin{equation}\label{CI}
\widehat{I}_{i,1-\alpha}^B = (\thh_i-z_{\alpha/2}\sh_i, \thh_i+z_{\alpha/2}\sh_i),
\end{equation}
where $\sh_i^2=\st_i^2(\Ah)$, and $z_{\alpha/2}$ is the upper $100(\alpha/2)\%$-quantile of the standard normal distribution. 
When the assumed normal model for $\theta_i$ in the model (\ref{model}) is approximately correct, both estimate $\thh_i$ and interval $\widehat{I}_{i,1-\alpha}^B$ give efficient estimation and inference results for $\theta_i$. 
However, in many applications, there may exist outlying areas having distinguished values of $\theta_i$ compared to the other areas, and the normality assumption is no more reasonable for such situations.

Before introducing the proposed method, we first reveal the properties of the (estimated) posterior means and variances given in (\ref{pos}) when there exist some representative outliers that cannot be captured by the assumed normal distribution for $\theta_i$.
To simplify the situation, we consider a situation where the absolute value of the first observation goes to infinity, namely, $|y_1|\to\infty$, while other observed values are fixed. 
Such framework is typically adopted to investigate the robustness of shrinkage estimation \citep[e.g.][]{carvalho2010horseshoe}. 
Under the framework, we can prove the following property:

\begin{prp}\label{Proposition:CI crude}
Assume that $m$ is fixed.
Then as $|y_1| \to \infty$, it holds that $\thh_i - y_i \to 0$ and $\sh_i^2 \to  D_i$ for all $i=1,\ldots,m$.
\end{prp}

The above result indicates that only a single outlying area affects the posterior means and variances of all the areas, and both posterior means and variances end up with the direct estimator $y_i$ and sampling variance $D_i$. 
This means that the shrinkage properties in (\ref{pos}) are disappeared by failing ``borrowing strength", leading to inefficient estimation and inference on non-outlying areas. 
In particular, the confidence interval (\ref{CI}) reduces to the direct interval
\begin{equation}\label{crude}
\widehat{I}^D_{i,1-\alpha} = (y_i-z_{\alpha/2}\sqrt{D_i},y_i+z_{\alpha/2}\sqrt{D_i}).
\end{equation}
Ideally, we should be able to construct efficient estimates and confidence intervals by ``borrowing strength" for non-outlying areas, while the direct estimate $y_i$ and interval $\widehat{I}^D_{i,1-\alpha}$  should be used for outlying areas.
However, the practical difficulty is that we do not know which area is outlying.
Our aim of this study is to provide an adaptive way to give robust estimation and inference of $\theta_i$ depending on the existence of outlying areas.

\subsection{Robust empirical Bayes confidence intervals}
We first note that the empirical Bayes confidence interval (\ref{CI}) is characterized by $\tht_i$ and $\st_i^2$ in (\ref{pos}). 
According to the Tweedie's formula \citep[e.g.][]{efron2011tweedie}, they can be derived by the form of marginal distribution:  
\begin{equation}\label{tw}
\tht_i=y_i+D_i\frac{d}{dy_i}\log m(y_i), \ \ \ \ \ \ 
\st_i^2=D_i+D_i^2\frac{d^2}{dy_i^2}\log m(y_i),
\end{equation}
where $m(y_i)$ is the marginal density of $y_i\sim N(x_i^\top\beta, A+D_i)$.
Moreover, the estimator $\psih$ is obtained through the marginal distribution of $y_i$, which indicates that the marginal likelihood of $y_i$ completely characterizes both estimates and confidence intervals.

When some outlying areas are included (i.e., the normal model for $\theta_i$ is misspecified), the marginal distribution of $y_i$ is also misspecified, which is the main reason for the undesirable property as demonstrated in Proposition \ref{Proposition:CI crude}. 
To overcome the problem, we replace the marginal log-likelihood with a generalized likelihood based on $\gamma$-divergence \citep{jones2001comparison} given by 
\begin{equation}\label{GD}
\begin{split}
\mathcal{D}_{\gamma}(y_i;\psi) 
&\equiv \frac{1}{\gamma}\phi(y_i;x_i^\top\beta, A+D_i)^{\gamma}\left\{\int \phi(t;x_i^\top\beta, A+D_i)^{1+\gamma} dt\right\}^{-\gamma/(1+\gamma)}\\
&=\frac{1}{\gamma}\phi(y_i;x_i^\top\beta, A+D_i)^{\gamma}c_{\gamma}(A),
\end{split}
\end{equation}
where $c_{\gamma}(A)=\{2\pi(A+D_i)\}^{\gamma^2/2(1+\gamma)}$, and $\phi(y;\mu, \sigma^2)$ is the density function of the normal distribution $N(\mu,\sigma^2)$ and $\gamma>0$ is a tuning parameter related to the robustness for which a selection procedure will be discussed later. 
Note that $\lim_{\gamma\to 0}\{\mathcal{D}_{\gamma}(y_i;\psi)-\gamma^{-1}\}=\log m(y_i)$, that is, the objective function is very similar to the marginal log-likelihood under small $\gamma$.
Hence, the objective function (\ref{GD}) is a natural generalization of the marginal log-likelihood. 
We note that $\gamma$-divergence has stronger robustness \citep[e.g.][]{fujisawa2008robust} against the existence of outlying areas, which plays an important role in the theoretical justification of the proposed confidence intervals in Section \ref{sec:theory}.

Our main idea is to replace $\log m(y_i)$ with $\mathcal{D}_{\gamma}(y_i)$ in the Tweedie's formula (\ref{tw}) to obtain robust posterior mean and variance.
A straightforward calculation shows that 
\begin{equation}\label{tw2}
\begin{split}
\tht_i^{(\gamma)}
&\equiv y_i+D_i\frac{d}{dy_i}\mathcal{D}_{\gamma}(y_i;\psi)=y_i-w_{\gamma}(y_i)\frac{D_i}{A+D_i}(y_i-x_i^\top\beta), \\
\st_i^{2(\gamma)}
&\equiv D_i+D_i^2\frac{d^2}{dy_i^2}\mathcal{D}_{\gamma}(y_i;\psi)=D_i+w_{\gamma}(y_i)\frac{D_i^2}{(A+D_i)^2}\left\{\gamma(y_i-x_i^\top\beta)^2- (A+D_i)\right\},
\end{split}
\end{equation}
where $w_{\gamma}(y_i)=\phi(y_i;x_i^\top\beta, A+D_i)^{\gamma}c_{\gamma}(A)$.
It can be easily checked that $\tht_i^{(\gamma)}\to \tht_i$ and $\st_i^{2(\gamma)}\to \st_i^2$ as $\gamma\to 0$, so that the expression in (\ref{tw}) can be seen as a robust generalization of (\ref{pos}).
We define the robust estimator of $\psi$ as 
\begin{equation}\label{est}
\psih^{(\gamma)}={\rm argmax}_{\psi}\sum_{i=1}^m\mathcal{D}_{\gamma}(y_i; \psi),
\end{equation}
noting that $\psih^{(\gamma)}$ reduces to the maximum likelihood estimator as $\gamma\to 0$. 
Since the above objective function is an analytic form of $\psi$, it can be easily optimized, for example, by the Newton-Raphson algorithm. 
Then, we propose a robust empirical Bayes confidence interval as 
\begin{equation}\label{EB-RCI}
\widehat{I}_{i,1-\alpha}^{(\gamma)} = (\thh_i^{(\gamma)}-z_{\alpha/2}\sh_i^{(\gamma)}, \thh_i^{(\gamma)}+z_{\alpha/2}\sh_i^{(\gamma)}),
\end{equation} 
where $\thh_i^{(\gamma)}$ and $\sh_i^{(\gamma)}$ are obtained by replacing $\psi$ with $\psih^{(\gamma)}$ in $\tht_i^{(\gamma)}$ and $\st_i^{(\gamma)}$, respectively. 
A notable property of the new interval (\ref{EB-RCI}) is that not only $\tht_i^{(\gamma)}$ but also $\st_i^{(\gamma)}$ depends on the observed value $y_i$, which results in the following quite different property from one given in Proposition \ref{Proposition:CI crude}.

\begin{prp}\label{Proposition:CI proposed}
Let $\gamma \in (0,1]$. Assume that $m$ is fixed.
As $|y_1|\to\infty$, it holds that $\thh_1^{(\gamma)} - y_1 \to 0$ and $\sh_1^{(\gamma)}\to \sqrt{D_1}$ while $\thh_i^{(\gamma)} \to \thh_{i,\dagger}^{(\gamma)}$ and $\sh_i^{(\gamma)} \to \sh_{i,\dagger}^{(\gamma)}$ for $i\neq 1$, where $\thh_{i,\dagger}^{(\gamma)}$ and $\sh_{i,\dagger}^{(\gamma)}$ are defined by replacing $\psih^{(\gamma)}$ with $\psih^{(\gamma)}_{\dagger}$, which is the maximizer of $m^{-1}\sum_{i=2}^{m}\mathcal{D}_{\gamma}(y_i;\psi)$. 
\end{prp}

\medskip
Proposition \ref{Proposition:CI proposed} implies two adaptive properties of the proposed confidence interval. 
For non-outlying signals, the proposed intervals produce efficient intervals by borrowing the strength of information from other non-outlying signals. 
In other words, the proposed interval can automatically identify the outlying subjects and successfully eliminate such data to provide a stable interval for non-outlying observations. 
The proposed interval reduces to the direct interval for outlying signals, which makes sense because there would be no information to borrow strength from non-outlying observations. 
Such property can be recognized as an extension of ``tail robustness" \citep{carvalho2010horseshoe} for the shrinkage estimation. 
Note that Proposition \ref{Proposition:CI proposed} merely demonstrates finite sample properties of the proposed confidence intervals as a function of the $m$-dimensional observation, $(y_1,\ldots,y_m)$, and the asymptotic coverage properties will be investigated in Section \ref{sec:theory}.

\subsection{Adaptation of the tuning parameter}\label{sec: gamma-selection}
The choice of $\gamma$ in (\ref{EB-RCI}) controls the trade-off between efficiency for non-outliers and robustness against outliers.
We here propose selecting a suitable value of $\gamma$ to obtain intervals with minimum posterior variances. 
Let $\Gamma\equiv \{\gamma_1,\ldots,\gamma_M\}$ be a set of candidate values for $\gamma$, where $0=\gamma_1<\gamma_2<\cdots<\gamma_M=1$. 
Then, we propose selecting $\gamma$ minimizing the total robust variances, that is, we define the optimal $\gamma$ as
\begin{equation}\label{gamma-selection}
\gamma_{op}={\rm argmin}_{\gamma\in\Gamma}\sum_{i=1}^ma_i\sh_i^{2(\gamma)},
\end{equation}
where $a_i$ is a fixed weight, for example, $a_i=1/D_i$.
We set $a_i=1$ in our theoretical argument for simplicity, but the extension is quite straightforward. 
We also note that the optimal value $\gamma_{op}$ provides information on whether the assumed normal distribution is plausible or not in terms of the efficiency of confidence intervals.

\section{Asymptotic properties}\label{sec:theory}

In Section \ref{sec:method} we discussed the properties of the proposed confidence intervals as a function of direct estimates, $y_1,\ldots,y_m$ with fixed $m$.
In this section, we will discuss the asymptotic properties of the proposed method under large numbers of areas.
In particular, we reveal selection performance for $\gamma$ based on (\ref{gamma-selection}) and asymptotic coverage probability of the robust confidence intervals under both correct specification and the existence of outlying areas. 
It should be noted that we assume that the first stage model, $y_i|\theta_i\sim N(\theta_i, D_i)$, is correctly specified, which is very likely in practice because $y_i$ is the direct estimator of $\theta_i$ and the normality assumption is approximately true. 
Throughout this section, we assume the following regularity conditions: 

\begin{condition}
$\{(y_i,x_i,D_i)\}$ is a sequence of independent and identically distributed random vectors and there exist positive constants $D_{L}$ and $D_{U}$ such that $0< D_{L} \leq D_i \leq D_{U}<\infty$ for $i=1,\ldots,m$.
\end{condition}

\begin{condition}
We have that (i) $x_i \in \mathcal{X}$, where $\mathcal{X} \subset \mathbb{R}^p$ is a bounded set and (ii) $E[x_i x_i^\top]$ a positive-definite matrix.
\end{condition}

\begin{condition}
$\Psi \subset \mathbb{R}^p \times (0,\infty)$ is a compact parameter space of $\psi_{\ast} \equiv  (\beta_{\ast},A_{\ast})$, where $\psi_{\ast}$ is the true parameter, and $A_{\ast} + D_L > (2\pi)^{-1}e^{-3}$.
\end{condition}

Uniform boundedness of $D_i$ in Condition 1 and positive definiteness in Condition 2 are widely adopted in the context of small area estimation \citep[e.g.][]{datta2005measuring,yoshimori2014second}.
The assumption that $D_i$ is a random variable in Condition 1 reflects the practical situation where $D_i$'s are estimated variance of $y_i$'s. 
Condition 3 is satisfied when the lower bound of $D_i$ is not very small.

\subsection{Properties under correct specification}\label{sec:correct}
We first consider the situation that the assumed model of $\theta_i$ is correctly specified, that is, there are no outlying areas, which is the standard framework for small area estimation.
The first result is the property of selecting $\gamma$ via (\ref{gamma-selection}). 

\begin{theorem}
\label{gamma_op}
Suppose that $\theta_i \sim N(x_i^{\top}\beta_{\ast}, A_{\ast})$, that is, the model (\ref{model}) is correctly specified. 
Then, under Conditions 1-3, it holds that $P(\gamma_{op} = 0) \to 1$ as $m\to\infty$.
\end{theorem}

Theorem \ref{gamma_op} indicates that if the distribution of $\theta_i$ is correctly specified, then the optimal $\gamma_{op}$ converges in probability to $0$ as $m \to \infty$. 
This implies that our method is asymptotically consistent with the standard parametric empirical Bayes confidence intervals when there are no outlying areas, which guarantees the efficiency of the proposed robust method.

We next consider asymptotic coverage probabilities of the proposed interval when the assumed model is correctly specified.

\begin{theorem}\label{coverage-conv}
When the model (\ref{model}) is correctly specified, it holds under Conditions 1-3 that
 \[
 P\left(\th_i \in \widehat{I}_{i,1-\alpha}^{(\gamma_{op})}\right) = 1-\alpha + o(1)\ \text{as $m \to \infty$}.
 \]
\end{theorem}

Theorem \ref{coverage-conv} implies that if the model is correctly specified, then our robust empirical Bayes confidence intervals using the optimal $\gamma_{op}$ selected via (\ref{gamma-selection}) have asymptotically correct coverage probabilities. 
Although we have not shown it in detail, the coverage accuracy of the proposed interval is roughly $O(m^{-1})$, which is the same rate as the standard confidence interval $\widehat{I}_{i,1-\alpha}^B$.
The coverage accuracy can improved by using, for instance, parametric bootstrap \citep[e.g.][]{chatterjee2008parametric,hall2006parametric} to make the coverage error $O(m^{-3/2})$ when $m$ is not very large.

\subsection{Properties under existence of outlying areas}\label{sec:cont}

We next consider a realistic situation under existence of outlying areas.
To this end, we suppose that $\theta_i$ are drawn from the density given by 
\begin{equation}\label{huber}
 f_{\psi_{\ast}}(\theta_i) = (1-\omega)\phi(\theta_i;x_i^\top \beta_{\ast},A_{\ast}) + \omega \delta(\theta_i). 
\end{equation}
where $\omega \in [0,1/2)$ is a contamination ratio and $\delta$ is a contamination distribution.
The data generating process is known as the Huber's contamination model \citep{huber1964robust,huber1965robust}.
We define the following quantity: 
\begin{align}\label{rho}
\rho &\equiv \max_{k\in \{0,1,2\}}\sup_{D \in [D_{L},D_{U}]} \sup_{x\in \mathcal{X}}\int f_{\delta}(y)\phi(y;x^\top \beta_{\ast},A_{\ast}+D)^{\gamma}|y-x^\top \beta_{\ast}|^k dy,
\end{align}
where $f_{\delta}(y)=\int \phi(y;\theta,D)\delta(\theta)d\theta$
is the marginal distribution of $y_i$ when $\theta_i\sim \delta$. 
Note that (\ref{rho}) measures separability between the genuine distribution $N(x_i^\top \beta_{\ast},A_{\ast})$ and contamination distribution $\delta$; $\rho$ is almost 0 when the two distributions are well separated, that is, the density $f_{\delta}$ mostly lies on the tail of the density $\phi(y_i;x_i^\top \beta_{\ast},A_{\ast}+D_i)$. 
Let $E_{f_{\psi}}[\cdot]$ be the expectation with respect to the model $y_i|\theta_i,D_i \sim N(\theta_i,D_i)$ that $\theta_i$ are drawn from the density  $f_{\psi}(\theta_i)$, and let $E_{\psi}[\cdot]$ and  $E_{\delta}[\cdot]$ denote the expectation with respect to $\phi(\theta_i;x_i^\top\beta,A+D_i)$ and $f_{\delta}$. 

Define $\partial_{\psi}\mathcal{D}_{\gamma}(y_i;\psi) = \partial\mathcal{D}_{\gamma}(y_i;\psi)/\partial\psi$ (see (\ref{GD}) for the definition of $\mathcal{D}_{\gamma}(y_i;\psi)$). Assume that $\gamma>0$ and let $\psi_{\dagger}^{(\gamma)} = (\beta_{\gamma,\dagger}, A_{\gamma, \dagger})$ be a root of $E_{f_{\psi_{\ast}}}[\partial_{\psi} \mathcal{D}_{\gamma}(y_1;\psi)]=0$, and let $\psih_{m}^{(\gamma)} = (\beh_{\gamma},\Ah_{\gamma})$ be a root of the estimating equation $m^{-1}\sum_{i=1}^{m}\partial_{\psi}\mathcal{D}_{\gamma}(y_i;\psi)=0$. As a result of Lemmas S4 and S5 in the Supplementary Material, we can show 
\begin{align*}
    \psih_{m}^{(\gamma)} &= \psi_{\dagger}^{(\gamma)} + O_{p}(m^{-1/2})
    = \psi_{\ast} + O(\omega \rho) + O_{p}(m^{-1/2}),
\end{align*}
which plays an important role in the following theorems.

We first show the selection performance of $\gamma$ under (\ref{huber}).  

\begin{theorem}\label{gamma-conv-contami-2}
Define 
$$
Q(\ga;\psi) \equiv  E_{\psi_{\ast}}\left[ w_{\ga}(y_i) \frac{D_i^2}{(A+D_i)^2} \left\{ \ga (y_i - x_i^\top \beta)^2 - (A + D_i) \right\}\right].
$$
Assume that the first element of $x_i$ is the constant term, $\theta_i\sim f_{\psi_{\ast}}(\theta_i)$ with $\omega \in (0,1/2)$, $E\left[ \frac{x_i x_i^\top}{a + D_i} \right]$ is a nonsingular matrix for any $a > 0$, $\mu_{\delta} \equiv  \int \theta_i \delta(\theta_i) d\theta_i < \infty$, $\tau_{\delta}^2 \equiv \int (\theta_i-\mu_{\delta})^2 \delta(\theta_i) d\theta_i < \infty$, and $E[(x_i^\top \beta_{\ast} - \mu_{\delta})^2 | D_i] \geq A_{\ast} + D_U$ for all $i$. 
Further, suppose that Conditions (a)-(e) in Lemma S4 hold and Conditions 1-3 are satisfied. 
If there exists $\ga_{\ast} \in \Gamma \setminus \{0\}$ such that $Q(0;\psi_{\ast}) > (1-\omega) Q(\ga_{\ast};\psi_{\ast})$ and $\rho(\psi_{\ast},\ga_{\ast},\delta) \equiv \rho$ is sufficiently small, then we have $P(\gamma_{op}>0)\to 1$ as $m\to\infty$. 
\end{theorem}

From Theorem \ref{gamma-conv-contami-2}, the probability that a positive value $\gamma_{op}$ is selected approaches 1 as $m\to\infty$ if the true distribution of $\theta_i$ is misspecified. The condition $Q(0;\psi_{\ast}) > (1-\omega) Q(\ga_{\ast};\psi_{\ast})$ implies that if $\omega$ is small, then $\gamma_{\ast}$ should be close to $0$. Practically, it is sufficient to take the number of candidates $M$ in $\Gamma$ sufficiently large. In order to make $\rho(\psi_{\ast},\gamma_{\ast},\delta)$ sufficiently small when $\gamma_{\ast}$ is close to $0$, the distribution $\delta$ should be well-separated from $N(x_i^\top\beta_{\ast},A_{\ast})$.
Note that the condition $(x_i^\top \beta_{\ast} - \mu_{\delta})^2 \geq A_{\ast} + D_U$ is also related to the separability.

Next, we provide the asymptotic property of the proposed confidence intervals.
In the following argument, we set $\rho_{\psi_{\ast}} \equiv \rho$ to indicate that $\rho$ depends on $\psi_{\ast} = (\beta_{\ast},A_{\ast})$.

\begin{theorem}\label{contami-coverage}
Suppose that Conditions (a)-(e) in Lemma S4 hold and Conditions 1-3 are satisfied.
Let $0<\underline{\gamma}<\overline{\gamma} < 1$. 
Assume that $\theta_i\sim f_{\psi_{\ast}}(\theta_i)$, $\omega \rho_{\psi_{\ast}} = O(m^{-\zeta})$ as $m \to \infty$ for some $\zeta \in (0,1/2)$, and $E_{\delta}[y_i^2]<\infty$. Then for $\gamma \in [\underline{\gamma},\overline{\gamma}]$, we have
\begin{align*}
P\left(\th_i \in \widehat{I}_{i,1-\alpha}^{(\gamma)}\right) = 1-\alpha + O\left((\omega\rho_{\psi_{\ast}})^{1/4}+\gamma\right)\ \text{as $m \to \infty$}.
 \end{align*} 
 \end{theorem}

The term $O((\omega \rho_{\psi_{\ast}})^{1/4})$ in the approximation error of the confidence interval $\widehat{I}_{i,1-\alpha}^{(\gamma)}$ comes from the bias for estimating $\psi_{\ast}$ under the contaminated model $\theta \sim f_{\psi_{\ast}}(\theta_i)$. 
If $\rho_{\psi_{\ast}}$ is small, that is, the contamination density $f_{\delta}(y_i)$ mostly lies on the tail of the density $\phi(y_i;x_i^\top \beta_{\ast}, A_{\ast}+D_i)$, then the term $O((\omega \rho_{\psi_{\ast}})^{1/4})$ can be small even if $\omega$ is not small. In this case, the observations could be heavily contaminated. When $\rho_{\psi_{\ast}}$ is not small, $O((\omega \rho_{\psi_{\ast}})^{1/4})$ can be small if the contamination ratio $\omega$ is small.  
The term $O(\gamma)$ reflects the effect of robustification using the $\gamma$-divergence for constructing $\widehat{I}_{i,1-\alpha}^{(\gamma)}$ instead of the standard confidence interval $\widehat{I}_{i,1-\alpha}^B$.
Hence, the term $O(\gamma)$ is likely to be positive since the posterior variance $\st_i^{(\gamma)}$ increases with $\gamma$ and the resulting interval has a wider length than the parametric interval. 
This means that the term $O(\gamma)$ is thought to work in the direction of increasing the coverage probability, that is, at least from the practical point of view, the coverage probability can be
\begin{align*}
    P\left(\theta_i \in \widehat{I}_{i,1-\alpha}^{(\gamma)}\right) &\geq (1-\omega)(1-\alpha) + \omega (1-\alpha) - O((\omega \rho_{\psi_{\ast}})^{1/4})\\
    & = 1-\alpha + O((\omega \rho_{\psi_{\ast}})^{1/4}).
\end{align*}
We give more detailed discussion on this issue in Section~S3 in the Supplementary Material.

\section{Simulation study}\label{sec:sim}
We compare the performance of the proposed robust method with existing methods through simulation studies. 
We consider the two-stage model: 
$$
y_i\sim N(\theta_i, D_i),  \ \ \ \theta_i=\beta_0+\beta_1 x_{i1} + \beta_2x_{i2} + \sqrt{A}u_i, \ \ \ i=1,\ldots,m,
$$
where $m=100$, $(\beta_0, \beta_1, \beta_2)=(0, -1, 1)$ and two choices of $A\in\{1, 0.5\}$.
The two auxiliary variables $x_{i1}$ and $x_{i2}$ are generated from the standard normal distribution and Bernoulli distribution with $0.5$ success probability, respectively. 
Throughout this study, we divide $m$ areas into five groups containing an equal number of areas and set the same value of $D_i$ within each group. The $D_i$ pattern of the groups is $(0.2, 0.6, 1.0, 1.4, 2.0)$.
We considered the following five scenarios for the distribution of $u_i$:
\begin{align*}
&\text{(i)} \ \ u_i\sim N(0, 1), \ \ \ \ \ 
\text{(ii)} \ \ u_i\sim LN(0, 1), \ \ \ \ 
\text{(iii)} \ \  u_i\sim C, \\
&\text{(iv)}  \ \ u_i\sim 0.95N(0, 1)+0.05N(10, 1),  \ \ \ \ \ 
\text{(v)} \ \   u_i\sim 0.9N(0, 1)+0.1N(10, 1), 
\end{align*}
where $LN(a,b)$ denotes the log-normal distribution with log-mean $a$ and log-variance $b$, and $C$ denotes the Cauchy distribution. 
Remark that scenario (i) corresponds to the situation where there are no outliers, and the assumed model is correctly specified to check the efficiency loss of the proposed method compared with the standard empirical Bayes confidence intervals. 
On the other hand, the underlying distribution of $\theta_i$ is quite different from a normal distribution in the other scenarios. 
In particular, signals in scenarios (iv) and (v) produce some outlying values, which would make the standard empirical Bayes approach inefficient, as shown in Proposition ~\ref{Proposition:CI crude}.

We first investigate the performance in terms of point estimation. 
In addition to the proposed method, denoted by GD, we employ the standard empirical Bayes (EB) estimator, robust estimator based on density power divergence (DPD) by \cite{sugasawa2020robust} and two robust estimators using Huber's $\psi$-function proposed by \cite{sinha2009robust} (denoted by SR) and \cite{ghosh2008influence} (denoted by GEB). 
We set $5\%$ inflation mean squared errors to choose the tuning parameter, and used the rule-of-thumb value $K=1.345$ for the Huber's $\psi$-function. Using $R=2000$ simulated datasets, we computed mean squared errors, $(mR)^{-1}\sum_{i=1}^m\sum_{r=1}^R (\thh_i^{(r)}-\th_i^{(r)})^2$, where $\thh_i^{(r)}$ and $\th_i^{(r)}$ are estimated and true values of $\th_i$ in the $r$th replication. 
The results are presented in Table~\ref{tab:sim-mse}.
First, under scenario (i), the optimal $\gamma$ in the proposed GD method was selected $0$ for all the replications, and the performance of GD and EB are identical, which is consistent with our theoretical result given in Theorem \ref{gamma_op}. 
On the other hand, the other robust methods tend to be inefficient under scenario (i). 
On the other hand, once outlying areas are included or the normal distribution is severely misspecified as in scenarios (ii)$\sim$(v), the optimal value of $\gamma$ tends to be strictly positive, as proved in Theorem \ref{gamma-conv-contami-2}.
It is also observed that EB performs poorly while the other robust methods provide more accurate estimates than EB.
In particular, the proposed GD method attains the minimum value in all the scenarios.

We next investigate the performance of empirical Bayes confidence intervals. 
We computed $95\%$ confidence intervals of $\theta_i$ using the proposed GD method and the standard empirical Bayes (EB), and the $95\%$ direct confidence interval (DR) given in (\ref{crude}). 
Using 2000 simulated datasets, we computed coverage probability, and the average length of $95\%$ confidence intervals averaged over all the subjects. 
The results are shown in Table \ref{tab:sim-cp}, showing that the three methods provide coverage probability around the nominal level (95\%) in all the scenarios. 
However, the average length of EB is close to that of the direct interval when the assumed model is misspecified in scenarios except for scenario (i). 
On the other hand, GD provides a more efficient interval length than EB even when the model is severely misspecified.

\begin{table}[!htbp]
\caption{Mean squared errors (MSE) of small area estimators under five scenarios of $\theta_i$ and two scenarios of $A\in \{0.5, 1\}$. 
\label{tab:sim-mse}
}
{\small 
\begin{center}
\begin{tabular}{ccccccccccccccccc}
 \hline
 && \multicolumn{5}{c}{Scenario ($A=1$)} && \multicolumn{5}{c}{Scenario ($A=0.5$)}\\
  method &  & (i) & (ii) & (iii) & (iv) & (v) &  &  (i) & (ii) & (iii) & (iv) & (v) \\
 \hline
GD &  & 0.487 & 0.688 & 0.896 & 0.620 & 0.674 &  & 0.328 & 0.598 & 0.799 & 0.465 & 0.525 \\
($\gamma_{\rm opt}$) &  & (0.00) &( 0.18) & (0.22) & (0.14) & (0.19) &  & (0.00) & (0.10) & (0.20) & (0.11) & (0.15) \\
EB &  & 0.487 & 0.945 & 1.063 & 0.906 & 0.974 &  & 0.328 & 0.789 & 1.029 & 0.894 & 0.968 \\
DPD &  & 0.504 & 0.726 & 0.912 & 0.691 & 0.850 &  & 0.338 & 0.598 & 0.821 & 0.582 & 0.740 \\
SR &  & 0.496 & 0.882 & 1.040 & 0.804 & 0.927 &  & 0.332 & 0.720 & 0.987 & 0.771 & 0.914 \\
GEB &  & 0.524 & 0.881 & 1.039 & 0.802 & 0.923 &  & 0.374 & 0.697 & 0.986 & 0.767 & 0.909 \\
\hline
\end{tabular}
\end{center}
}
\end{table}

\begin{table}[!htbp]
\caption{ 
Empirical coverage probability (CP) and average length (AL) of the three $95\%$ confidence intervals under five scenarios of $\theta_i$ and two scenarios of $A\in \{0.5, 1\}$.
\label{tab:sim-cp}
}
{\small 
\begin{center}
\begin{tabular}{ccccccccccccccc}
 \hline
 &&& \multicolumn{5}{c}{Scenario ($A=1$)} && \multicolumn{5}{c}{Scenario ($A=0.5$)}\\
 & method &  & (i) & (ii) & (iii) & (iv) & (v) &  &  (i) & (ii) & (iii) & (iv) & (v) \\
 \hline
 & EB &  & 93.6 & 94.8 & 95.0 & 94.6 & 94.8 &  & 92.0 & 94.8 & 95.0 & 94.6 & 94.8 \\
CP & GD &  & 93.6 & 95.8 & 95.5 & 96.3 & 96.2 &  & 92.0 & 95.6 & 95.7 & 96.6 & 96.6 \\
 & DR &  & 95.0 & 95.0 & 95.1 & 95.0 & 95.0 &  & 95.0 & 95.0 & 95.1 & 95.0 & 95.0 \\
 \hline
 & EB &  & 2.54 & 3.57 & 3.82 & 3.47 & 3.62 &  & 2.03 & 3.22 & 3.78 & 3.44 & 3.61 \\
AL & GD &  & 2.54 & 3.20 & 3.59 & 3.11 & 3.23 &  & 2.03 & 2.95 & 3.42 & 2.74 & 2.90 \\
 & DR &  & 3.86 & 3.86 & 3.86 & 3.86 & 3.86 &  & 3.86 & 3.86 & 3.86 & 3.86 & 3.86 \\
\hline
\end{tabular}
\end{center}
}
\end{table}

\section{Example: small area estimation of average crime numbers in Tokyo}\label{sec:app}
We demonstrate the proposed methods using a dataset of the number of police-recorded crimes in the Tokyo metropolitan area, obtained from the University of Tsukuba and publicly available online (``GIS database of the number of police-recorded crime at O-aza, chome in Tokyo, 2009-2017'', available at \url{https://commons.sk.tsukuba.ac.jp/data_en}). 
This study focuses on the number of violent crimes in $2,855$ local towns in the Tokyo metropolitan area.
The number of violent crimes per year is observed for nine years (from 2009 to 2017) in each town. 
We first excluded areas where no violent crime has occurred in the nine years, resulting in $m=2826$ areas used in the analysis.  
We first computed the average number of crimes per unit km$^2$ (denoted by $Y_i$) and its sampling variances (denoted by $D_i$) for $i=1,\ldots,m$.
The dataset contains five area-level auxiliary information, entire population density (PD), day-time population density (DPD), the density of foreign people (FD), percentage of single-person households (SH), and average year of living (AYL).
Letting $x_i$ be a six-dimensional vector of intercept and the five covariates, we consider the Fay-Herriot model (\ref{model}), where $\theta_i$ is the true average number of crimes (per unit km$^2$) in this study.

We first applied the proposed GD method, where we searched the optimal value of $\gamma$ from 
$\gamma\in \{0, 0.005, 0.01,\ldots, 0.295, 0.300\}$.
The selection criterion (\ref{gamma-selection}) with weight $\omega_i=1/D_i$ for each value of $\gamma$ is shown in the left panel in Figure~\ref{fig:crime}.
We can observe that the criterion has a unique minimizer and found that $\gamma_{opt}=0.095$ (bounded away from 0), indicating that there exist outlying areas and the normality assumption for the distribution of $\theta_i$ seems to be violated in some areas.  
The estimated values of the regression coefficients and the variance parameter ($A$) based on the GD method with the optimal value of $\gamma$ are reported in Table \ref{tab:crime-coef}, where the estimates from the empirical Bayes (EB) and density power divergence (DPD) methods used in Section~\ref{sec:sim} are also reported for comparison.  
From Table \ref{tab:crime-coef}, it can be seen that the estimates of $\beta$ and $A$ in EB tend to be larger than those in GD because of the existence of outlying areas. 
In particular, the estimate of $A$ in EB is 20 times larger than that in GD, and such a large value of $A$ leads to inefficient interval estimation.
Still, the estimates by DPD seem to be slightly affected by outlying areas. 
In the right panel in Figure~\ref{fig:crime}, we present scatter plots of the estimates of area-wise average crime numbers produced by GD, EB, DPD, and GEB, against the direct estimator. 
This plot reveals the over-shrinkage property of EB, that is, resulting estimates in areas with large values of $y_i$ (i.e. outlying areas) tend to be much smaller than $y_i$. 
On the other hand, the three robust methods tend to produce relatively similar estimates in outlying areas.
In particular, GD provides almost identical estimates in outlying areas and shrunken estimates in non-outlying areas, showing adaptively robustness according to the values of $y_i$.

We also computed $95\%$ confidence intervals of $\theta_i$ based on GD and EB methods.
To see the efficiency of the interval estimation, we calculated scaled average length, $m^{-1}\sum_{i=1}^mL_i/\sqrt{D_i}$, where $L_i$ is the interval length in the $i$th area, which are reported in Table \ref{tab:crime-coef}.
The result shows that GD provides more efficient intervals than EB on average. 
Furthermore, we see a more detailed difference between GD and EB methods together with the direct interval (DR) of the form (\ref{crude}). 
To see this, we show $95\%$ confidence intervals in areas having the top 30 largest $y_i$ (i.e., outlying areas), in the left panel of Figure \ref{fig:crime-CI}, which reveals the undesirable properties of EB; not only the interval length but also the interval location of EB is much smaller than the direct estimator $y_i$.
On the other hand, the confidence intervals of GD are almost identical to those of DR, owing to the robustness property given in Proposition \ref{Proposition:CI proposed}.
Furthermore, as shown in the right panel in Figure \ref{fig:crime-CI}, GD confidence intervals in non-outlying areas are much shorter than DR by successful ``borrowing strength" from other areas.

\begin{figure}[!htb]
\centering
\includegraphics[width=14cm,clip]{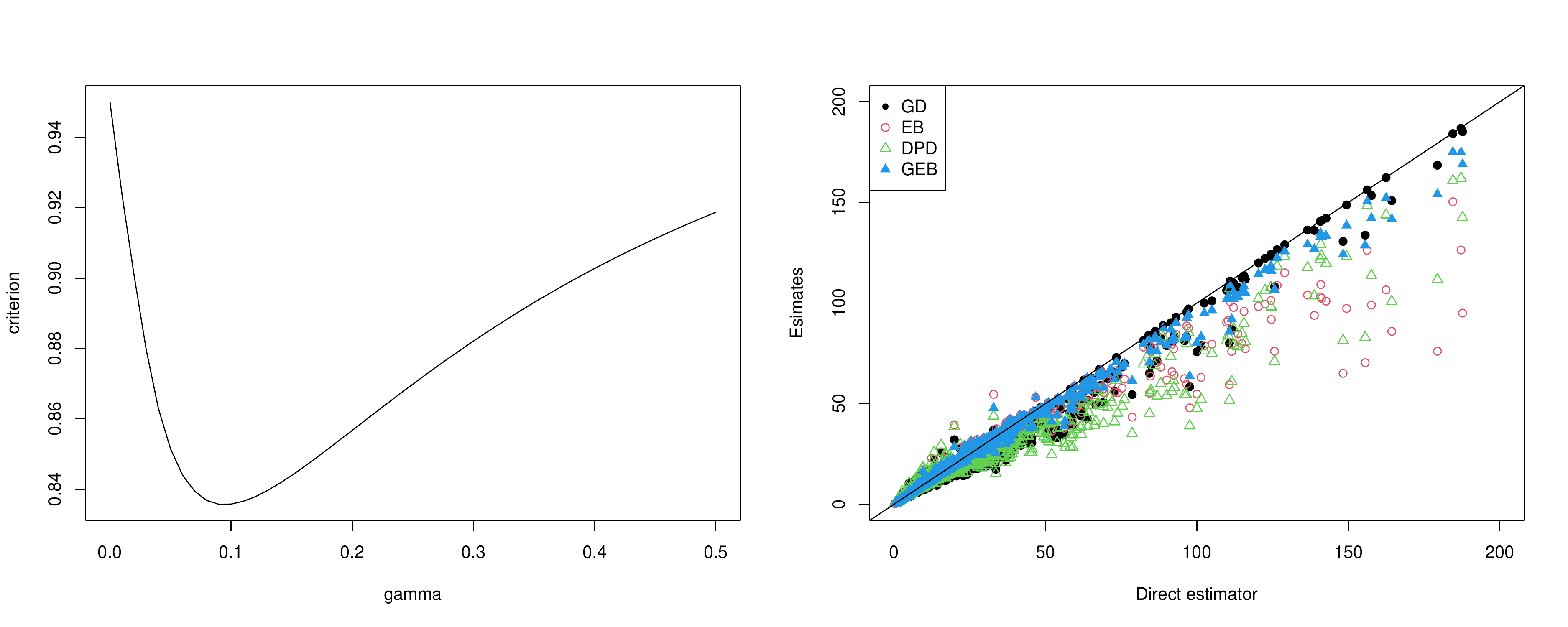}
\caption{Selection criterion as a function of $\gamma$ in the proposed GD method (left) and a scatter plot of direct estimates and empirical Bayes estimates of GD, EB, DPD and GEB methods used in Section \ref{sec:sim}, obtained in the analysis of crime data. }
\label{fig:crime}
\end{figure}

\begin{table}[htbp]
\caption{Estimates of the regression coefficients $\beta$ and variance parameter $A$, and scaled average interval lengths scaled by $\sqrt{D_i}$, of the proposed GD and the standard EB method. }
\label{tab:crime-coef}
\begin{center}
\begin{tabular}{cccccccccccccc}
\hline
& \multicolumn{6}{c}{Regression coefficient ($\beta$)}  & & Average\\
method & Int & PD & DPD & FD & SH & AYL &  $A$ & length\\
\hline
GD & 7.81 & 1.42 & 4.49 & 1.28 & 0.79 & 0.35 & 11.65 & 3.57 \\
EB & 12.20 & 0.60 & 8.64 & 2.96 & 2.64 & 1.64 & 231.54  & 3.81 \\
DPD & 9.33 &	1.53 &	5.52 &	1.65 &	1.32 & 0.67 &38.24 & --\\
\hline
\end{tabular}
\end{center}
\end{table}

\begin{figure}[!htb]
\centering
\includegraphics[width=14cm,clip]{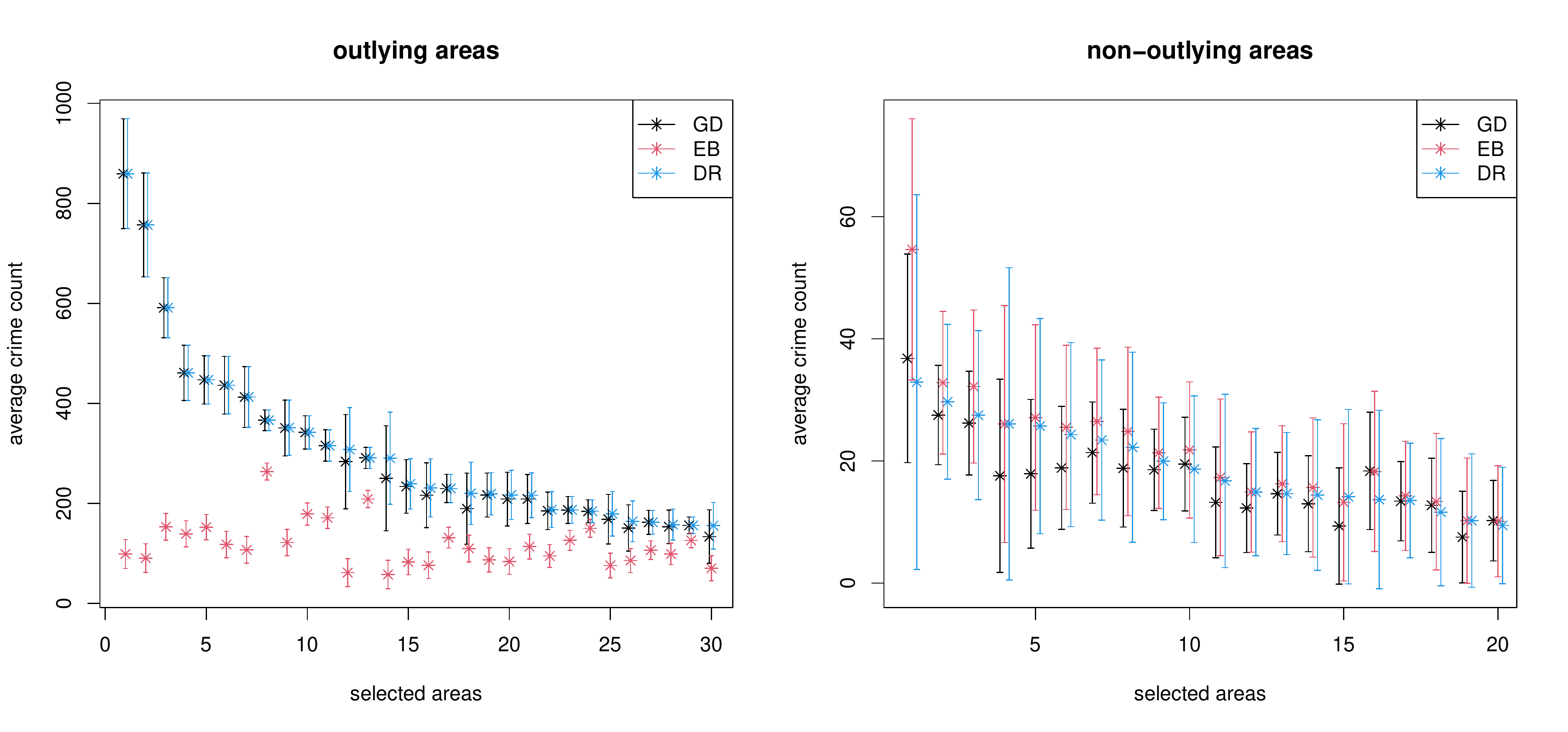}
\caption{$95\%$ confidence intervals in outlying areas having top 30 largest $y_i$ (left) and in 30 randomly selected non-outlying areas (right), obtained by GD, EB and DR methods}
\label{fig:crime-CI}
\end{figure}

\section{Discussion}\label{sec:dis}

Although we employed the $\gamma$-divergence as a robust alternative of the standard likelihood, it would be possible to adopt other robust divergences such as density power divergence \citep{basu1998robust} and $\alpha\beta$-divergence \citep{cichocki2011generalized}, to obtain a robust confidence interval.
The main reason to use $\gamma$-divergence is its strong robustness against contamination; that is, the robust estimator converges to a quantity with a small (latent) bias against the true value, which was the basis of our theoretical development.
The numerical advantage of such strong robustness can also be confirmed in our simulation study in Section~\ref{sec:sim}.

An important remaining work would be a further investigation of the second-order refinement of the proposed robust confidence intervals by using Bootstrap, as partly mentioned in Section~\ref{sec:method}.
Although the second-order refinement of the standard confidence interval can be achieved under correct specification, it would be challenging to extend the theory to the situation under the existence of outlying areas, as considered in Section~\ref{sec:cont}.
Moreover, it would also be challenging to extend the asymptotic approximation of mean squared errors \citep[e.g.][]{prasad1990estimation,das2004mean} to the situation under the existence of outlying areas.

\section*{Acknowledgement}
This work is partially supported by the Japan Society for Promotion of Science (KAKENHI) grant numbers 20K13468 and 21H00699.

\section*{Supplementary material}
The Supplementary Material includes proofs of Propositions \ref{Proposition:CI crude} and \ref{Proposition:CI proposed} (Supplementary Material S1), proofs of Theorems \ref{gamma_op}, \ref{coverage-conv},  \ref{gamma-conv-contami-2}, and \ref{contami-coverage} (Supplementary Material S2) and further discussion regarding the results of Theorem~\ref{contami-coverage} and additional simulation results on the effect of the tuning parameter $\gamma$ (Supplementary Material S3).

\bibliographystyle{chicago}
\bibliography{refs}



%
\newpage 

\setcounter{page}{1}
\setcounter{equation}{0}
\renewcommand{\theequation}{S\arabic{equation}}
\setcounter{section}{0}
\renewcommand{\thesection}{S\arabic{section}}
\setcounter{lemma}{0}
\renewcommand{\thelemma}{S\arabic{lemma}}
\setcounter{table}{0}
\renewcommand{\thetable}{S\arabic{table}}
\setcounter{figure}{0}
\renewcommand{\thefigure}{S\arabic{figure}}

\vspace{2cm}
\begin{center}
{\bf {\LARGE Supplementary Material for ``Adaptively Robust Small Area Estimation: Balancing Efficiency and Robustness of Empirical Bayes Confidence Intervals"}}
\end{center}

This supplementary material includes proofs of Propositions \ref{Proposition:CI crude} and \ref{Proposition:CI proposed} (Supplementary Material \ref{Proof: Prop1-2}), proofs of Theorems \ref{gamma_op}, \ref{coverage-conv},  \ref{gamma-conv-contami-2}, and \ref{contami-coverage} (Supplementary Material \ref{Proof: Thm1-4}) and additional simulation results on the effect of the tuning parameter $\gamma$ (Supplementary Material \ref{sec:sim-add}).

\section{Proofs for Section \ref{sec:method}}\label{Proof: Prop1-2}

\subsection{Proof of Proposition \ref{Proposition:CI crude}}

Note that the maximum likelihood estimators $\beh$ and $\Ah$ based on the marginal likelihood (\ref{ML}) solve the equations
\begin{align}
    \sum_{i=1}^{m}\frac{x_iy_i}{A + D_i} &= \sum_{i=1}^{m}\frac{x_ix_i^\top \beta}{A + D_i}, \label{MLE-eq1}\\ 
    \sum_{i=1}^{m}\frac{(y_i - x_i^\top \beta)^2}{(A + D_i)^2} &= \sum_{i=1}^{m}{1 \over A+D_i}. \label{MLE-eq2}
\end{align}
From (\ref{MLE-eq1}), we have
\begin{align*}
\max_{1\leq i \leq m}\left| {x_i^\top \beh \over x_i^\top \tilde{\beta}} - 1 \right| = O\left(\frac{1}{|y_1|}\right) \to 0\ \text{as $|y_1| \to \infty$},
\end{align*}
where $\tilde{\beta} = \big\{\sum_{i=1}^{m}x_ix_i^\top(A + D_i)^{-1}\big\}^{-1}x_1y_1/(A+D_1)$. Replacing $\beta$ with $\tilde{\beta}$ in (\ref{MLE-eq2}), we have
\begin{align*}
&\sum_{i=1}^{m}\frac{(y_i -  x_i^\top\tilde{\beta})^2}{(A + D_i)^2}\\ 
&= \frac{y_1^2}{(A + D_1)^2}\left(1- x_1^\top\left(\sum_{i=1}^{m}\frac{1}{A + D_i}\right)^{-1}\left(\frac{x_1}{A + D_1}\right)\right)^2\\
    &\quad + \frac{y_1^2}{(A+D_i)^2}\sum_{i=2}^{m}\left(x_i^\top \left(\sum_{j=1}^{m}{x_j x_j^\top \over A + D_j}\right)^{-1}x_1\right)^2\frac{1}{(A + D_i)^2}\\
    &\quad - \frac{2y_1}{A + D_1}\sum_{i=2}^{m}\left(x_i^\top \left(\sum_{j=1}^{m}{x_j x_j^\top \over A + D_j}\right)^{-1}x_1\right)\frac{1}{(A + D_i)^2} +  \sum_{i=2}^{m}\frac{y_i^2}{(A + D_i)^2}\\
    &= O\left({|y_1|^2 \over A^2}\right)
\end{align*}
as $|y_1| \to \infty$. In this case, (\ref{MLE-eq2}) implies $\sum_{i=1}^{m}{1 \over A + D_i} = O\left({|y_1|^2 \over A^2}\right)$. Since the left hand side is $O(A^{-1})$ and the right hand side is $O(|y_1|^2A^{-2})$, we have $\Ah = O(|y_1|^2)$. Therefore,  
\begin{align*}
\thh_i - y_i &=- \frac{D_i}{\Ah + D_i}(y_i - x_i^\top\beh)\\
&=- \frac{D_i}{\Ah + D_i}\left\{y_i - x_i^\top \left(\sum_{j=1}^{m}\frac{x_j x_j^\top}{A + D_i}\right)^{-1}\left(\frac{x_1y_1}{A+D_1}\right)(1+o(1))\right\} \to 0
\end{align*}
and $\sh_i^2 = \frac{\Ah D_i}{\Ah + D_i} - D_i \to 0$ as $|y_1| \to \infty$. 

\subsection{Proof of Proposition \ref{Proposition:CI proposed}}

Define $\mathcal{D}_{\gamma,y_1}(\psi)\equiv  m^{-1}\sum_{i=1}^{m}\mathcal{D}_{\gamma}(y_i;\psi)$.
Let $\psih_{\dagger}^{(\gamma)} = (\beh^{(\gamma)}_{\dagger},\Ah^{(\gamma)}_{\dagger})$ be the maximizer of 
\begin{align*}
\mathcal{D}_{\gamma,\infty}(\psi)\equiv  \frac{1}{m}\sum_{i=2}^{m}\mathcal{D}_{\gamma}(y_i;\psi) 
\end{align*}
and let $\Psi_{\infty} \subset \mathbb{R}^p \times (0,\infty)$ be a compact set such that $\psih_{\dagger}^{(\gamma)} \in \Psi_{\infty}$. 

First we prove $\psih^{(\gamma)} \to \psih^{(\gamma)}_{\dagger}$ as $|y_1| \to \infty$. Note that 
$\psih^{(\gamma)} \in \Psi_{\infty}$ for sufficiently large $|y_1|$, and as $|y_1| \to \infty$, $\sup_{\psi \in \Psi_{\infty}}|\mathcal{D}_{\gamma,y_1}(\psi) - \mathcal{D}_{\gamma,\infty}(\psi)| \to 0$. Then we have 
\begin{align}
    |\mathcal{D}_{\gamma,\infty}(\psih^{(\gamma)}_{\dagger})-\mathcal{D}_{\gamma,\infty}(\psih^{(\gamma)})| &\leq |\mathcal{D}_{\gamma,\infty}(\psih^{(\gamma)}_{\dagger})-\mathcal{D}_{\gamma,y_1}(\psih^{(\gamma)}_{\dagger})| + |\mathcal{D}_{\gamma,y_1}(\psih^{(\gamma)}_{\dagger})-\mathcal{D}_{\gamma,y_1}(\psih^{(\gamma)})| \nonumber \\ 
    &\quad + |\mathcal{D}_{\gamma,y_1}(\psih^{(\gamma)})-\mathcal{D}_{\gamma,\infty}(\psih^{(\gamma)})| \nonumber\\
    &\leq 2\sup_{\psi \in \Psi_{\infty}}|\mathcal{D}_{\gamma,y_1}(\psi) - \mathcal{D}_{\gamma,\infty}(\psi)| + |\mathcal{D}_{\gamma,y_1}(\psih^{(\gamma)}_{\dagger})-\mathcal{D}_{\gamma,y_1}(\psih^{(\gamma)})| \nonumber \\
    &\to 0\ \text{as $|y_1| \to \infty$}. \label{GD-MLE-fin-sample}
\end{align}
Since $\mathcal{D}_{\gamma,\infty}(\psi)$ is continuous over $\psi \in \Psi_{\infty}$, (\ref{GD-MLE-fin-sample}) implies that 
\begin{align}\label{conv-MLE-dagger}
    \psih^{(\gamma)} \to \psih^{(\gamma)}_{\dagger}\ \text{as $|y_1| \to \infty$}. 
\end{align}

Next we prove (i) $\thh_1^{(\gamma)} - y_1 \to 0$ and $\sh_1^{(\gamma)}\to D_1$, and (ii) $\thh_i^{(\gamma)} \to \thh_{i,\dagger}^{(\gamma)}$ and $\sh_i^{(\gamma)} \to \sh_{i,\dagger}^{(\gamma)}$ for $i=2,\hdots,m$ as $|y_1| \to \infty$. It is straightforward to show (ii) by applying (\ref{conv-MLE-dagger}). Since $|y_1|^{2}\phi(y_1:x_1^\top \beh^{(\gamma)}_{\dagger}, \Ah^{(\gamma)}_{\dagger} + D_1)^{\gamma} \to 0$ as $|y_1| \to \infty$, we have $\thh_1^{(\gamma)} - y_1 \to 0$ and $\sh_1^{(\gamma)}- \sqrt{D_1} \to 0$.

\section{Proofs for Section \ref{sec:theory}}\label{Proof: Thm1-4}

We first prove the results on $\gamma_{op}$ (Theorems \ref{gamma_op} and \ref{gamma-conv-contami-2}) and then prove the results on confidence intervals (Theorems \ref{coverage-conv} and \ref{contami-coverage}). 

\subsection{Proof of Theorem \ref{gamma_op}}

Before we see the asymptotic property of $\gamma_{op}$, we prepare some auxiliary results. Our approach is based on general theory for $M$-estimators \citep[e.g.][]{vanderVaart1998}. Lemma \ref{unif_LLN} is a uniform law of large numbers for a class of functions. Lemma \ref{conv-mu_tau} provides the consistency of $\psih^{(\gamma)}$.

\begin{lemma}\label{unif_LLN}
Let $\{(y_i,x_i,D_i)\}_{i \geq 1}$ be a sequence of i.i.d. random vectors such that $0 < D_L \leq D_i \leq D_U < \infty$ and $E[y_i^2], E[\|x_i\|^2] < \infty$. Define $\mathcal{F} = \{ (y,x,d) \mapsto f_{\psi, \ga}(y,x,d) : \psi \in \Psi, \ga \in [0,1] \}$ and
\begin{align*}
    f_{\psi,\ga} (y,x,d) &= \phi(y:x^\top\beta,A + d)^{\ga} \frac{d^2}{(A + d)^2} \\
    & \hspace{0.8in} \times \left\{ \ga (y_1-x^\top\beta)^2 - (A + d) \right\} \{2\pi (A + d)\}^{\ga^2 / 2(1+\ga)}
\end{align*}
where $\psi = (\beta, A)$ and $\Psi$ is a compact subset of $\mathbb{R}^{p} \times (0,\infty)$. Then we have
\[
\sup_{f \in \mathcal{F}} \left| P_m f - P f \right| \stackrel{p}{\to}0 \ \text{as $m \to \infty$},
\]
where $P_{m}f = \frac{1}{m}\sum_{i=1}^{m}f(y_i,x_i,D_i)$ and $Pf = E[f(y_i,x_i,D_i)]$.
\end{lemma}
\begin{proof}
Because $0 < D_L \leq D_i \leq D_U < \infty$, for $\psi \in \Psi$ and $\ga \in [0,1]$, there exists a positive constant $C$ such that
\[
\left| f_{\psi,\ga} (y_i,x_i,D_i) \right| \ \leq \ C(y_i - x_i^\top\beta)^2.
\]
Because $\Psi$ is compact and $E[y_i^2], E[\|x_i\|^2] < \infty$, there exists a function $F$ such that $\sup_{\psi \in \Psi, \ga \in [0,1]} |f_{\psi,\ga} (y_i,x_i,D_i)| \leq F(y_i,x_i,D_i)$ and $E[F(y_i,x_i,D_i)] < \infty$. In addition, the map $(\psi, \ga) \mapsto f_{\psi,\ga}(y,x,d)$ is continuous for every $(y,x,d)$ and $\Psi \times [0,1]$ is compact. Hence, from Example 19.8 in \cite{vanderVaart1998}, $\mathcal{F}$ is Glivenko-Cantelli.
\end{proof}

\begin{lemma}\label{conv-mu_tau}
Let $\{(y_i,x_i,D_i)\}_{i \geq 1}$ be a sequence of i.i.d. random vectors and there exist positive constants $D_{L}$ and $D_{U}$ such that $0< D_{L} \leq D_i \leq D_{U}<\infty$. Assume that $\Psi \subset \mathbb{R}^{p} \times (0,\infty)$ is a compact parameter space of $\psi_{\ast} \equiv  (\beta_{\ast},A_{\ast})$. When $y_i | \theta_i, D_i \sim N(\theta_i, D_i)$ and $\theta_i \sim N(x_i^\top\beta_{\ast}, A_{\ast})$, for any $\ga \in [0,1]$ we have
\[
\beh_{\ga} \stackrel{p}{\to} \beta_{\ast} \ \ \text{and} \ \ \Ah_{\gamma} \stackrel{p}{\to} A_{\ast}.
\]
\end{lemma}
\begin{proof}
For any $\gamma \in (0,1]$, we define
\begin{align*}
    g_{\ga}(y_i,x_i,D_i;\psi) &\equiv  \frac{1}{\ga} \phi(y_i;x_i^\top \beta,A+D_i)^{\ga} c_{\gamma}(A) - \frac{1}{\gamma}.
\end{align*}
In addition, we define $g_{0}(y_i,x_i,D_i;\psi) \equiv  \log \left\{ \phi(y_i;x_i^\top \beta,A+D_i) \right\}$. Then, we have $\psih^{(\ga)} = (\beh_{\ga},\Ah_{\ga}) = {\rm argmax}_{\psi \in \Psi} m^{-1} \sum_{i=1}^m g_{\ga}(y_i,x_i,D_i;\psi)$. Similar to Lemma \ref{unif_LLN}, from Example 19.8 in \cite{vanderVaart1998}, for any $\ga \in [0,1]$ we have
\begin{equation}
    \sup_{\psi \in \Psi} \left| \frac{1}{n} \sum_{i=1}^m g_{\ga}(y_i,x_i,D_i;\psi) - E[g_{\ga}(y_i,x_i,D_i;\psi)] \right| = o_p(1).
    \label{D-conv}
\end{equation}
We observe that
\begin{align*}
\frac{\partial}{\partial \beta} &E[g_{\ga}(y_i,x_i,D_i;\psi)] = E\left[ E \left[ \phi(y_i;x_i^\top\beta,A+D_i)^{\ga} \left( \frac{y_i-x_i^\top\beta}{A+D_i} \right)c_{\gamma}(A) x_i \Big| x_i, D_i \right] \right], \\
\frac{\partial}{\partial A}
&E[g_{\ga}(y_i,x_i,D_i;\psi)]\\
& \ \  
= \frac{1}{2} E\left[ E\left[ \frac{\phi(y_i;x_i^\top\beta,A+D_i)^{\ga}}{(A+D_i)^2}  \left\{ (y_i-x_i^\top\beta)^2 - \frac{A+D_i}{1+\gamma} \right\}c_{\gamma}(A)\ \Big| x_i, D_i \right] \right].
\end{align*}
Since $\phi(y_i;x_i^\top\beta , A+D_i)$ is symmetric around $y_i=x_i ^\top \beta$, the first derivative is $0$ when $\beta=\beta_{\ast}$ and $A=A_{\ast}$.
Regarding the second derivative, note that when $\beta=\beta_{\ast}$ and $A=A_{\ast}$, 
\begin{equation}\label{expect}
\begin{split}
E[\phi(y_i;x_i^\top \beta,A+D_i)^{\ga}|x_i,D_i]&=V_i(1+\gamma)^{-1/2}\\
E[\phi(y_i;x_i^\top \beta,A+D_i)^{\ga}(y_i-x_i^\top \beta)^2|x_i,D_i]&=V_i(1+\gamma)^{-3/2}(A+D_i), 
\end{split}
\end{equation}
where $V_i=\{2\pi(A+D_i)\}^{-\gamma/2}$.
This shows that the second derivative is also $0$. 
This implies that $E[g_{\ga}(y_i,D_i;\psi)]$ is maximized at $\psi = (\beta_{\ast},A_{\ast})$. Hence, from Theorem 5.7 in \cite{vanderVaart1998}, we obtain consistency of $\beh_{\ga}$ and $\Ah_{\ga}$ for all $\ga \in [0,1]$.
\end{proof}

Now we move on to the proof of Theorem \ref{gamma_op}. We define
\begin{align*}
\hat{Q}_{m}(\ga; \beta, A) &\equiv   \frac{1}{n}\sum_{i=1}^{m}\phi(y_i;x_i^\top \beta,A+\si_i^2)^{\gamma}\frac{D_i^2}{(A+D_i)^2}\left\{\ga(y_i-x_i^\top \beta)^2- (A+D_i)\right\}c_{\gamma}(A), \\
Q(\ga; \beta, A) &\equiv  E\left[ \phi(y_i;x_i^\top \beta,A+\si_i^2)^{\gamma}\frac{D_i^2}{(A+D_i)^2}\left\{\ga(y_i-x_i^\top \beta)^2- (A+D_i)\right\}c_{\gamma}(A) \right],
\end{align*}
then $\ga_{op}={\rm argmin}_{\ga \in \Gamma}\hat{Q}_{m}(\ga;\beh_{\ga},\Ah_{\gamma})$. 
By the definition of $\psih_{\ga} \equiv  (\beh_{\ga},\Ah_{\gamma})$, we have $(\beh_{\ga}, \Ah_{\gamma}) \in \Psi$ with probability 1. From the triangular inequality, we obtain
\begin{align*}
    & \left| \hat{Q}_{m}(\ga;\beh_{\ga},\Ah_{\gamma}) - \hat{Q}_{m}(\ga;\beta_{\ast},A_{\ast}) \right| \\
    &\leq \left| \hat{Q}_{m}(\ga;\beh_{\ga},\Ah_{\gamma}) - Q(\ga;\beh_{\ga},\Ah_{\gamma}) \right| + \left| \hat{Q}_{m}(\ga;\beta_{\ast},A_{\ast}) - Q(\ga;\beta_{\ast},A_{\ast}) \right| \\
    & \hspace{0.5in} + \left| Q(\ga;\beh_{\ga},\Ah_{\gamma}) - Q(\ga;\beta_{\ast},A_{\ast}) \right| \\
    & \leq 2 \sup_{(\beta,A) \in \Psi} \left| \hat{Q}_{m}(\ga;\beta,A) - Q(\ga;\beta,A) \right| + \left| Q(\ga;\beh_{\ga},\Ah_{\gamma}) - Q(\ga;\beta_{\ast},A_{\ast}) \right|.
\end{align*}
Because $\beh_{\ga} \stackrel{p}{\to} \beta_{\ast}$ and $\Ah_{\gamma} \stackrel{p}{\to} A_{\ast}$ from Lemma \ref{conv-mu_tau} and $Q(\ga;\beta,A)$ is continuous in $\beta$ and $A$, it follows from Lemma \ref{unif_LLN} that $\hat{Q}_{m}(\ga;\beh_{\ga},\Ah_{\gamma}) = \hat{Q}_{m}(\ga;\beta_{\ast},A_{\ast}) + o_p(1) = Q(\ga;\beta_{\ast},A_{\ast}) + o_p(1)$.

From (\ref{expect}), we have 
$$
Q(\ga;\beta_{\ast}, A_{\ast})
=-E\left[\frac{D_i^2}{A_{\ast}+D_i}\{2\pi(A_{\ast}+D_i)\}^{-\gamma/2(1+\gamma)}(1+\gamma)^{-3/2} \right].
$$
Then, it follows that 
\begin{align*}
\frac{\partial}{\partial\gamma}&Q(\gamma; \beta_{\ast},A_{\ast}) \\
&=\frac12 E\left[\frac{D_i^2}{A_{\ast}+D_i}\{2\pi(A_{\ast}+D_i)\}^{-\gamma/2(1+\gamma)}(1+\gamma)^{-5/2}\left\{\frac{\log(2\pi(A_{\ast}+D_i))}{1+\gamma}+3\right\}\right]\\
&\geq \frac{\sigma_{L}^{4}}{2(A_{\ast} + D_U)}\{2\pi(A_{\ast}+D_U)\}^{-\gamma/2(1+\gamma)}
(1+\gamma)^{-5/2}\left\{\frac{\log(2\pi(A_{\ast}+D_L))}{1+\gamma}+3\right\}.
\end{align*}
Then, under the condition $A_{\ast} + D_L > (2\pi)^{-1}e^{-3}$, we have  
\begin{align*}
\frac{\partial}{\partial \gamma}Q(\gamma; \beta_{\ast},A_{\ast})>0\ \text{for $\gamma \in [0,1]$}. 
\end{align*}
This implies that $Q(\gamma; \beta_{\ast},A_{\ast})$ is strictly increasing on $\gamma \in [0,1]$. Hence, $\hat{Q}_{m}(0;\beh_{0},\Ah_{0}) < \min_{\ga \in \Gamma \setminus \{0\}} \hat{Q}_{m}(\ga;\beh_{\ga},\Ah_{\gamma})$ holds with probability approaching 1. This implies that $P(\ga_{op} = 0) \to 1$.

\subsection{Proof of Theorem \ref{gamma-conv-contami-2}}

Fix $\omega \in (0,1)$ and $\delta$. From Lemmas \ref{lemma: psudo-psi-conv} and \ref{lemma: latent-bias-bound}, we obtain $\hat{\psi}^{(\gamma)}_m = \psi_{\dagger}^{(\gamma)} + O_p(m^{-1/2})$ and $\psi_{\dagger}^{(\gamma)} = \psi_{\ast} + O\left(\omega \rho(\psi_{\ast},\ga,\delta) \right)$ for any $\ga \in \Gamma \setminus \{0\}$. Similar to Lemma \ref{conv-mu_tau}, for any $\gamma \in \Gamma$ we have
$$
\sup_{\psi \in \Psi} \left| \hat{Q}_m(\ga;\psi) - \tilde{Q}(\ga;\psi) \right| = o_p(1),
$$
where $\tilde{Q}(\ga;\psi) \equiv  E_{f_{\psi_{\ast}}}\left[ \phi(y_i;x_i^\top \beta,A+D_i^2)^{\gamma}\frac{D_i^2}{(A+D_i)^2}\left\{\ga(y_i-x_i^\top \beta)^2- (A+D_i)\right\}c_{\gamma}(A) \right]$. For any $\ga \in \Gamma$, we obtain
$$
\hat{Q}_{m}\left( \ga; \hat{\psi}_m^{(\ga)} \right)  \to_p \tilde{Q} \left(\ga; \psi_{\dagger}^{(\gamma)} \right).
$$
Because $\tilde{Q}(\ga;\psi)$ is differentiable in $\psi$ and $\psi_{\dagger}^{(\gamma)} = \psi_{\ast} + O\left(\omega \rho(\psi_{\ast},\ga,\delta) \right)$, $\tilde{Q} \left(\ga; \psi_{\dagger}^{(\gamma)} \right) = \tilde{Q} \left(\ga; \psi_{\ast} \right) + O\left(\omega \rho(\psi_{\ast},\ga,\delta) \right)$ holds for any $\ga \in \Gamma \setminus \{0\}$.

Next, we consider $\ga = 0$. Let $\psi_{\dagger}^{(0)} \equiv (\beta_{\dagger}, A_{\dagger})$ be the probability limit of $\hat{\psi}_m^{(0)} \equiv (\beh_{0}, \Ah_{0})$ under $\theta_i \sim f(\theta_i)$. Similar to the above argument, we have
$$
\hat{Q}_m(0;\hat{\psi}_m^{(0)}) = - \frac{1}{m} \sum_{i=1}^m \frac{D_i^2}{\hat{A}_{0} +D_i} \stackrel{p}{\to} - E \left[  \frac{D_i^2}{A_{\dagger} +D_i} \right].
$$
Then, $(\beta_{\dagger}, A_{\dagger})$ satisfies the following equations:
\begin{align*}
    E_{f_{\psi_{\ast}}}\left[ \frac{x_i (y_i - x_i^\top \beta_{\dagger})}{A_{\dagger} + D_i} \right] = 0 \ \ \text{and} \ \ E_{f_{\psi_{\ast}}}\left[ \frac{(y_i- x_i \top \beta_{\dagger})^2 - (A_{\dagger} + D_i) }{(A_{\dagger} + D_i)^2} \right] = 0.
\end{align*}
Because $E_{f_{\psi_{\ast}}}[y_i|x_i, D_i] = (1-\omega) x_i^\top \beta_{\ast} + \omega x_i^\top v_{\mu_{\delta}}$ where $v_{\mu_{\delta}} \equiv (\mu_{\delta}, 0, \ldots, 0)^\top$, the first equality implies
\begin{align*}
    & \ \ 0 = E\left[ \frac{x_i x_i^\top \{(1-\omega) \beta_{\ast} + \omega v_{\mu_{\delta}} - \beta_{\dagger}\}}{A_{\dagger} + D_i} \right] \\
    \Leftrightarrow & \ \ \beta_{\dagger} = (1-\omega)\beta_{\ast} + \omega v_{\mu_{\delta}},
\end{align*}
where the second equation follows from the non-singularity of $E\left[ \frac{x_i x_i^\top}{A_{\dagger} + D_i} \right]$. Then, we obtain
\begin{align*}
    & E_{f_{\psi_{\ast}}}\left[ \frac{(y_i- x_i^\top \beta_{\dagger})^2}{(A_{\dagger} + D_i)^2} \Big| x_i, D_i \right] \\
    &= \frac{(1-\omega)\left\{ (A_{\ast} + D_i) + \{x_i^\top( \beta_{\dagger} - \beta_{\ast})\}^2 \right\} + \omega \left\{ (\tau_{\delta}^2 + D_i) + \{x_i^\top (\beta_{\dagger} - v_{\mu_{\delta}})\}^2 \right\}  }{(A_{\dagger} + D_i)^2} \\
    &= \frac{(1-\omega)(A_{\ast} + D_i) + \omega (\tau_{\delta}^2 + D_i) + \left\{ \omega^2 + (1-\omega)^2 \right\} \{x_i^\top (\beta_{\ast} - v_{\mu_{\delta}})\}^2 }{(A_{\dagger} + D_i)^2}.
\end{align*}
Combined with the second equation, $A_{\dagger}$ must satisfy the following equation:
$$
E_{f_{\psi_{\ast}}} \left[ \frac{(1-\omega)(A_{\ast} + D_i) + \omega (\tau_{\delta}^2 + D_i) + \left\{ \omega^2 + (1-\omega)^2 \right\} \{x_i^\top (\beta_{\ast} - v_{\mu_{\delta}})\}^2 - (A_{\dagger} + D_i) }{(A_{\dagger} + D_i)^2} \right] = 0.
$$
Because $D_L \leq D_i \leq D_U$ and $\omega \in (0,1/2)$, we have
\begin{align*}
    & (1-\omega)(A_{\ast} + D_i) + \omega (\tau_{\delta}^2 + D_i) + \left\{ \omega^2 + (1-\omega)^2 \right\} \{x_i^\top (\beta_{\ast} - v_{\mu_{\delta}})\}^2 - (A_{\dagger} + D_i) \\
    = & \ \omega (\tau_{\delta}^2 + D_i) + (A_{\ast} - A_{\dagger}) + \left\{ 2 (\omega - 1/2)^2 + 1/2 \right\} \{x_i^\top (\beta_{\ast} - v_{\mu_{\delta}})\}^2 - \omega (A_{\ast} + D_i) \\
    \geq & \ \omega D_L + (A_{\ast} - A_{\dagger}) + \frac{1}{2} \{x_i^\top (\beta_{\ast} - v_{\mu_{\delta}})\}^2 - \frac{1}{2} (A_{\ast} + D_U).
\end{align*}
Hence, we obtain
\begin{align*}
    & 0 \geq E\left[ (A_{\dagger} + D_i)^{-2} \right] \cdot \left\{ \omega D_L + (A_{\ast} - A_{\dagger}) \right\} + \frac{1}{2} E\left[ \frac{(x_i^\top \beta_{\ast} - \mu_{\delta})^2 - (A_{\ast} + D_U)}{(A_{\dagger} + D_i)^2} \right] \\
    \Leftrightarrow \ \ & A_{\dagger} - A_{\ast} \geq \omega D_L,
\end{align*}
where we use $E[(x_i^\top \beta_{\ast} - \mu_{\delta})^2 | D_i] \geq A_{\ast} + D_U$. As a result, we obtain
$$
\tilde{Q}\left( 0;\psi_{\dagger}^{(0)} \right) \equiv  - E\left[ \frac{D_i^2}{A_{\dagger} +D_i} \right] > \tilde{Q}(0;\psi_{\ast}) \equiv  - E\left[ \frac{D_i^2}{A_{\ast} +D_i} \right].
$$
We note that the difference is strictly positive uniformly with respect to $\delta$ that satisfies $E[(x_i^\top \beta_{\ast} - \mu_{\delta})^2 | D_i] \geq A_{\ast} + D_U$.

From the definition of $\rho(\psi_{\ast},\ga,\delta)$, for any $\ga \in \Gamma \setminus \{0\}$ we observe that
\begin{align*}
    \tilde{Q}(\ga;\psi_{\ast}) &= E_{f_{\psi_{\ast}}}\left[ w_{\ga}(y_i) \frac{D_i^2}{(A_{\ast}+D_i)^2}\left\{\ga(y_i-x_i^\top \beta_{\ast})^2- (A_{\ast}+D_i)\right\}c_{\gamma}(A_{\ast}) \right] \\
    &= (1-\omega) E_{\psi_{\ast}}\left[ w_{\ga}(y_i)\frac{D_i^2}{(A_{\ast}+D_i)^2}\left\{\ga(y_i-x_i^\top \beta_{\ast})^2- (A_{\ast}+D_i)\right\}c_{\gamma}(A_{\ast}) \right] \\
    & \hspace{2.0in} + O\left(\omega \rho(\psi_{\ast},\ga,\delta) \right) \\
    &= (1-\omega) Q(\ga;\psi_{\ast}) + O\left(\omega \rho(\psi_{\ast},\ga,\delta) \right),
\end{align*}
where $E_{\psi_{\ast}}$ denotes the expectation under $y_i | x_i,D_i \sim N(x_i^\top \beta_{\ast},A_{\ast} + D_i)$. As discussed in the proof of Theorem \ref{gamma_op}, $Q(\ga;\psi_{\ast})$ is continuous in $\ga$ and $Q(\ga;\psi_{\ast}) > Q(0;\psi_{\ast}) = \tilde{Q}(0;\psi_{\ast})$. As a result, we obtain
\begin{align*}
    \tilde{Q} \left(\ga_{\ast}; \psi_{\dagger}^{(\gamma_{\ast})} \right) &= \tilde{Q} \left(\ga_{\ast}; \psi_{\ast} \right) + O\left(\omega \rho(\psi_{\ast},\ga_{\ast},\delta) \right) = (1-\omega) Q \left(\ga_{\ast}; \psi_{\ast} \right) + O\left(\omega \rho(\psi_{\ast},\ga_{\ast},\delta) \right) \\
    & < Q \left(0; \psi_{\ast} \right) + O\left(\omega \rho(\psi_{\ast},\ga_{\ast},\delta) \right) < \tilde{Q}\left( 0;\psi_{\dagger}^{(0)} \right) + O\left(\omega \rho(\psi_{\ast},\ga_{\ast},\delta) \right).
\end{align*}
Hence, if $\rho(\psi_{\ast},\ga_{\ast},\delta)$ is sufficiently small, then we have $P(\gamma_{op}>0)\to 1$ as $m\to\infty$.

\subsection{Proof of Theorem \ref{coverage-conv}}

The following lemma (Nazarov's inequality) plays an important role to show the asymptotic validity of the proposed confidence intervals. 

\begin{lemma}\label{Nazarov-ineq}
(\cite{Nazarov2003}; Lemma A.1 in \cite{CCK2017}).
Let $Y = (Y_1,\hdots,Y_d)^{\top}$ be a centered Gaussian random vector in $\mathbb{R}^d$ such that $E[Y_{i}^{2}] \geq \si_0^2$ for all $i=1,\hdots, p$ and some constant $\si_0>0$. Then for every $y \in \mathbb{R}^d$ and $\delta>0$, 
\begin{align*}
    P(Y \leq y + \delta) - P(Y \leq y) &\leq \frac{\delta}{\si_0}(\sqrt{2\log d} + 2). 
\end{align*}
\end{lemma}

Now we move on to the proof of Theorem \ref{coverage-conv}. Recall that $\widehat{I}_{i,1-\alpha}^{(\gamma)} = (\thh_i^{(\gamma)}-z_{\alpha/2}\sh_i^{(\gamma)}, \thh_i^{(\gamma)}+z_{\alpha/2}\sh_i^{(\gamma)})$. Then $\widehat{I}_{i,1-\alpha}^{(0)} = \widehat{I}_{i,1-\alpha}^B$ (see (\ref{CI}) for the definition of $\widehat{I}_{i,1-\alpha}^B$). Since $P(\gamma_{op} = 0) \to 1$ as $m \to \infty$, we have
\begin{align*}
    P\left(\theta_i\in \widehat{I}_{i,1-\alpha}^{(\gamma_{op})}\right) &= P\left(\left\{\theta_i\in \widehat{I}_{i,1-\alpha}^{(\gamma_{op})}\right\} \cap \left\{\gamma_{op} = 0\right\}\right) + P\left(\left\{\theta_i\in \widehat{I}_{i,1-\alpha}^{(\gamma_{op})}\right\} \cap \left\{\gamma_{op} > 0\right\}\right)\\
    &\leq P\left(\theta_i\in \widehat{I}_{i,1-\alpha}^B\right) + P(\gamma_{op}>0) = P\left(\theta_i\in \widehat{I}_{i,1-\alpha}^B\right) + o(1). 
\end{align*}
Likewise, we have $P\left(\theta_i\in \widehat{I}_{i,1-\alpha}^{(\gamma_{op})}\right) \geq P\left(\theta_i\in \widehat{I}_{i,1-\alpha}^B\right) - o(1)$.

Define $\widetilde{C}_{i,1-\alpha}$ by replacing $\thh_{i}$ and $\sh_i$ with $\tht_i$ and $\st_i$, respectively where $\tht_i$ and $\st_i$ are defined in (\ref{pos}). Note that the standard posterior is given as $\theta_i|y_i,x_i,D_i \sim N(\tht_i,\st_i^2)$. Then
\begin{align*}
    P\left(\theta_{i} \in \widetilde{C}_{i,1-\alpha}\right) &= P\left(\left|\st_i^{-1}(\th_i-\tht_i)\right| \leq z_{\alpha/2}\right) = E\left[P\left(\left. \left|\st_i^{-1}(\th_i-\tht_i)\right| \leq z_{\alpha/2}\right|y_i,x_i, D_i\right)\right]\\
    &= E[\Phi(z_{\alpha/2}) - \Phi(z_{1-\alpha/2})] =1-\alpha, 
\end{align*}
where $\Phi(\cdot)$ is the cumulative distribution function of the standard normal distribution.
Now we will show $P\left(\th_i \in \widehat{I}_{i,1-\alpha}^B\right) = P\left(\th_i \in \widetilde{C}_{i,1-\alpha}\right) + o(1)$. 
Since $\Ah \stackrel{p}{\to} A_{\ast}$ and $\beh \stackrel{p}{\to} \beta_{\ast}$ as $m \to \infty$ from Lemma \ref{conv-mu_tau}, there exists a sequence of constants $\epsilon_{m} \to 0$ such that $P(\Omega_{m}) \to 1$  where $\Omega_{m} = \{\max\{|\Ah - A_{\ast}|, \| \beh - \beta_{\ast}\| \} \leq \epsilon_{m}\}$. Observe that 
 \begin{align}
     &P\left(\th_i \leq \thh_{i} + \sh_i z_{\alpha/2}\right) \nonumber \\ 
     &= P\left(\th_i \leq \tht_i + \st_i z_{\alpha/2} + (\thh_{i} - \tht_i) + (\sh_i-\st_i) z_{\alpha/2}\right) \nonumber \\
     &= P\left(\left\{\th_i \leq \tht_i + \st_i z_{\alpha/2} + (\thh_{i} - \tht_i) + (\sh_i-\st_i) z_{\alpha/2}\right\}\cap \Omega_{m} \cap \{|y_{i} - x_i^\top \beta_{\ast}| \leq \epsilon_{m}^{-1/2}\}\right) \nonumber \\
     &\quad + P\left(\left\{\th_i \leq \tht_i + \st_i z_{\alpha/2} + (\thh_{i} - \tht_i) + (\sh_i-\st_i) z_{\alpha/2}\right\}\cap \Omega_{m} \cap \{|y_{i} - x_i^\top \beta_{\ast}| > \epsilon_{m}^{-1/2}\}\right) \nonumber \\
     &\quad + P\left(\left\{\th_i \leq \tht_i + \st_i z_{\alpha/2} + (\thh_{i} - \tht_i) + (\sh_i-\st_i) z_{\alpha/2}\right\}\cap \Omega_{m}^{c}\right) \nonumber \\
     &\leq P\left(\left\{\th_i \leq \tht_i + \st_i z_{\alpha/2} + (\thh_{i} - \tht_i) + (\sh_i-\st_i) z_{\alpha/2}\right\}\cap \Omega_{m} \cap \{|y_{i} - x_i^\top \beta_{\ast}| \leq \epsilon_{m}^{-1/2}\}\right) \nonumber \\
     &\quad + P(|y_{i} - x_i^\top \beta_{\ast}| > \epsilon_{m}^{-1/2}) + P(\Omega_{m}^{c}) \nonumber \\
     &= P\left(\left\{\th_i \leq \tht_i + \st_i z_{\alpha/2} + (\thh_{i} - \tht_i) + (\sh_i-\st_i) z_{\alpha/2}\right\}\cap \Omega_{m} \cap \{|y_{i} - x_i^\top \beta_{\ast}| \leq \epsilon_{m}^{-1/2}\}\right) + o(1). \label{thh-bonund-coverage}
 \end{align}
On $\Omega_{m} \cap \{|y_{i} - x_i^\top \beta_{\ast}| \leq \epsilon_{m}^{-1/2}\}$, we have 
 \begin{align}
     |\thh_i - \tht_i| &\leq \left|\frac{D_i}{\Ah + D_i} - \frac{D_i}{A_{\ast} + D_i}\right||y_i - x_i^\top \beta_{\ast}| + \frac{D_i}{A_{\ast} + D_i}|x_i^\top \beh - x_i^\top \beta_{\ast}| \nonumber \\
     &\quad + \left|\frac{D_i}{\Ah + D_i} - \frac{D_i}{A_{\ast} + D_i}\right||x_i^\top \beh - x_i^\top \beta_{\ast}| \nonumber \\
     &\leq \frac{D_U\epsilon_m}{(A_{\ast}+D_{L}-\epsilon_m)^{2}}\epsilon_{m}^{-1/2} + \frac{D_{U}}{A_{\ast}+D_{L}}\epsilon_m + \frac{D_U\epsilon_m^{2}}{(A_{\ast}+D_{L}-\epsilon_m)^{2}} \nonumber \\
     &\leq C_{1}(D_L, D_U,A_{\ast})\epsilon_{m}^{1/2},
     \label{th-bound-coverage}
 \end{align}
and 
\begin{align}
    \left|\sh_i-\st_i\right| &\leq \left|\sh_i^2-\st_i^2\right|\frac{1}{|\sh_i+\st_i|}\nonumber \\
    &\leq \left(\frac{D_i|\Ah - A_{\ast}|}{A_{\ast} + D_i} + D_iA_{\ast}\left|\frac{1}{\Ah + D_i} - \frac{1}{A_{\ast} + D_i}\right| \right. \nonumber \\
    &\left. \quad + D_i\left|\Ah - A_{\ast}\right|\left|\frac{1}{\Ah + D_i} - \frac{1}{A_{\ast} + D_i}\right|\right)\frac{1}{2|\st_i|-|\sh_i-\st_i|} \nonumber \\
    &\leq \left(\frac{D_U\epsilon_m}{A_{\ast} + D_L} + \frac{D_UA_{\ast}\epsilon_{m}}{(A_{\ast} + D_L-\epsilon_m)^{2}} + \frac{D_U\epsilon_{m}^{2}}{(A_{\ast} + D_L-\epsilon_m)^{2}}\right)\frac{1}{2|\st_i|-|\sh_i-\st_i|} \nonumber \\
    &\leq C_{2}(D_L, D_U,A_{\ast})\epsilon_m \label{si-bound-coverage}
\end{align}
for some positive constants $C_{1}(D_L, D_U,A_{\ast})$ and $C_{2}(D_L, D_U,A_{\ast})$. 
Combining (\ref{th-bound-coverage}), (\ref{si-bound-coverage}), and Lemma \ref{Nazarov-ineq}, we have
 \begin{align}
    &P\left(\left. \left\{\th_i \leq \tht_i + \st_i z_{\alpha/2} + (\thh_{i} - \tht_i) + (\sh_i-\st_i) z_{\alpha/2}\right\}\cap \Omega_{m} \cap \{|y_{i} - x_i^\top \beta_{\ast}| \leq \epsilon_{m}^{-1/2}\}\right|y_i,x_i,D_i\right) \nonumber \\
    &\leq P\left(\left.\th_i \leq \tht_i + \st_i z_{\alpha/2} + C_{1}(D_L, D_U,A_{\ast})\epsilon_{m}^{1/2} + C_{2}(D_L, D_U,A_{\ast})\epsilon_m z_{\alpha/2}\right|y_i,x_i,D_i\right) \nonumber \\
    &\leq P\left(\left.\th_i \leq \tht_i + \st_i z_{\alpha/2} \right|y_i,x_i,D_i\right) + C_{3}(D_L, D_U,A_{\ast})\epsilon_{m}^{1/2} \nonumber \\
    & = 1-\alpha/2 + C_{3}(D_L, D_U,A_{\ast})\epsilon_{m}^{1/2} \label{tht-bound-coverage}
 \end{align}
 for some positive constant $C_{3}(D_L, D_U,A_{\ast})$. Combining (\ref{thh-bonund-coverage}) and (\ref{tht-bound-coverage}), we have $P\left(\th_i \leq \thh_{i} + \sh_i z_{\alpha/2}\right) \leq 1-\alpha/2 + o(1)$. Likewise, we have $P\left(\th_i \leq \thh_{i} + \sh_i z_{\alpha/2}\right) \geq 1-\alpha/2 - o(1)$. Almost the same argument yields $P\left(\th_i \leq \thh_{i} - \sh_i z_{\alpha/2}\right) = \alpha/2 + o(1)$. Therefore, we have
 \begin{align*}
     P\left(\th_i \in \widehat{I}_{i,1-\alpha}^{(\gamma_{op})}\right) = 1-\alpha + o(1). 
 \end{align*}
 
 \subsection{Proof of Theorem \ref{contami-coverage}}

To investigate the asymptotic properties of the confidence interval $\widehat{I}_{i,1-\alpha}^{(\gamma)}$, we first prepare two auxiliary results. From Lemma \ref{lemma: psudo-psi-conv}, $\psih^{(\gamma)}$ converge to a "psudo-true" parameter $\psi_{\dagger}^{(\gamma)}$, and from Lemma \ref{lemma: latent-bias-bound}, $\psi_{\dagger}^{(\gamma)}$ has a bias $O(\omega\rho_{\psi_{\ast}})$ with respect to the true parameter $\psi_{\ast}$.

\begin{lemma}\label{lemma: psudo-psi-conv}
(Theorem 5.41 and 5.42 of \cite{vanderVaart1998}; modified version of Theorem 4.1 of \cite{Fujisawa2013}).
Suppose that $\gamma > 0$ and $z_1=(y_1,x_1,D_1),\hdots, z_m=(y_m,x_m,D_m)$ are i.i.d samples from the model $y_i|\theta_i,D_i \sim N(\theta_i,D_i)$, $\theta_i \sim f_{\psi^\ast}(\theta_i)$ and that there exist positive constants $D_{L}$ and $D_{U}$ such that $0< D_{L} \leq D_i \leq D_{U}<\infty$.
We assume: (a) The function $\partial_{\psi} \mathcal{D}_{\gamma}(z;\psi)$ is twice continuously differentiable with respect to $\psi$ for any $z = (y,x,D) \in \mathbb{R}\times \mathcal{X} \times [D_L,D_U]$. (b) There exists a root $\psi_{\dagger}^{(\gamma)}$ such that $E_{f_{\psi_{\ast}}}[\partial_{\psi} \mathcal{D}_{\gamma}(z;\psi)]=0$. (c) $E_{f_{\psi_{\ast}}}\left[\|\partial_{\psi} \mathcal{D}_{\gamma}(z;\psi)\|^{2}\right]<\infty$. (d) $E_{f_{\psi_{\ast}}}\left[\partial (\partial_{\psi} \mathcal{D}_{\gamma}(z;\psi_{\dagger}^{(\gamma)}))/\partial \psi^{\top}\right]$ exists and is nonsingular. (e) The second-order differentials of $\partial_{\psi} \mathcal{D}_{\gamma}(z;\psi)$ with respect to $\psi$ are dominated by a fixed integrable function $\overline{D}(z)$ in a neighborhood of $\psi = \psi_{\dagger}^{(\gamma)}$. Then there exists a sequence of roots $\{\psih_{m}^{(\gamma)}\}_{m=1}^{\infty}$, such that 
\begin{itemize}
    \item[(i)] $\psih_{m}^{(\gamma)} \stackrel{p}{\to} \psi_{\dagger}^{(\gamma)}$,
    \item[(ii)] $\sqrt{m}(\psih_{m}^{(\gamma)} - \psi_{\dagger}^{(\gamma)}) \stackrel{d}{\to} N(0, \tau_f^{2}(\psi_{\dagger}^{(\gamma)}))$,
\end{itemize}
where $\tau_f^{2}(\psi) = J_{f}(\psi)K_f(\psi)\{J_f(\psi)^{\top}\}^{-1}$, $J_f(\psi)=E_{f_{\psi_{\ast}}}\left[\partial (\partial_{\psi} \mathcal{D}_{\gamma}(x;\psi))/\partial \psi^{\top}\right]$, $K_f(\psi)=E_{f_{\psi_{\ast}}}[\partial_{\psi} \mathcal{D}_{\gamma}(z;\psi)\partial_{\psi} \mathcal{D}_{\gamma}(z;\psi)^{\top}]$. 
\end{lemma}

Assumptions (a)-(e) in Lemma \ref{lemma: psudo-psi-conv} can be verified for the estimators based on $\gamma$-divergence (see \cite{Fujisawa2013} for example).

In the following results, we suppose that the maximization of the objective function in (\ref{est}) is performed on the set $\Psi_{m} = \{\psi: \rho_{\psi}^{1/\gamma} \leq K\rho_{\psi^{\ast}}^{1/\gamma}\}$ such that $\psi_{\ast}, \psi_{\dagger}^{(\gamma)} \in \Psi_{m}$, where $K$ is a positive constant and $\psi_{\ast} = ( \beta_{\ast},A_{\ast})$ and $\rho_{\psi}$ is defined in the same way as $\rho_{\psi_{\ast}}$ by replacing $\psi_{\ast}$ with $\psi=(\beta,A)$.  
\begin{lemma}\label{lemma: latent-bias-bound}
(Theorem 3.3 of \cite{fujisawa2008robust}).
Suppose that $\psi_{\ast}, \psi_{\dagger}^{(\gamma)} \in \Psi_m$. Then it holds that $\psi_{\dagger}^{(\gamma)} = \psi_{\ast} + O(\omega \rho_{\psi_{\ast}})$.  
\end{lemma}

From Lemmas \ref{lemma: psudo-psi-conv} and \ref{lemma: latent-bias-bound}, we have  
\begin{align*}
   \psih_{m}^{(\gamma)} = \psi_{\ast} + O\left(\omega \rho_{\psi_{\ast}}\right) + O_{p}(m^{-1/2})
\end{align*}
and this enables us to evaluate coverage probabilities of the proposed confidence intervals $\widehat{I}_{i,1-\alpha}^{(\gamma)}$. 

Now we move on to the proof of Theorem \ref{contami-coverage}. Since 
  \begin{align*}
      \psih^{(\gamma)} = (\beh_{\gamma}, \Ah_{\gamma}) = \psi_{\ast} + O(\omega\rho_{\psi_{\ast}}) + O_{p}(1/\sqrt{m})
  \end{align*}
  as $m \to \infty$ from Lemmas \ref{lemma: psudo-psi-conv} and \ref{lemma: latent-bias-bound}, there exists a sequence of constants $\tilde{\epsilon}_{m} \to 0$ such that $\tilde{\epsilon}_m = m^{-\zeta/2} + (\omega \rho_{\psi_{\ast}})^{1/2} = O(m^{-\zeta/4})$ and $P(\Omega_{m}) \to 1$ as $m \to \infty$ where $\Omega_{m} = \{\max\{|\Ah_{\gamma} - A_{\ast}|, \|\beh_{\gamma} - \beta_{\ast}\|\} \leq \tilde{\epsilon}_{m}\}$. Observe that $\theta_i$ can be represented as $\theta_i = (1-B_i)G_i + B_i\Delta_i$, where $B_i$, $G_i$, and $\Delta_i$ are i.i.d. random variables such that $B_i$ and $\{G_i,\Delta_i\}$ are independent, $P(B_i=1)=\omega$ and $P(B_i=0)=1-\omega$, $G_i|x_i \sim N(x_i^\top \beta_{\ast},A_{\ast})$, and $\Delta_i$ is a random variable with the density function $\delta$. Likewise, $y_i$ can be represented as $y_i=\theta_i+\sqrt{D_i}Z_i$, where $Z_i \sim N(0,1)$. 
  From the same argument to show (\ref{thh-bonund-coverage}), we have  
 \begin{align*}
     &P\left(\th_i \leq \thh_{i}^{(\gamma)} + \sh_i^{(\gamma)} z_{\alpha/2}\right)  \\ 
     &= P\left(\left\{\th_i \leq \tht_i^{(\gamma)} + \st_i^{(\gamma)} z_{\alpha/2} + (\thh_{i}^{(\gamma)} - \tht_i^{(\gamma)}) + (\sh_i^{(\gamma)}-\st_i^{(\gamma)}) z_{\alpha/2}\right\}\cap \Omega_{m} \cap \{|y_{i} - x_i^\top\beta_{\ast}| \leq \tilde{\epsilon}_{m}^{-1/2}\}\right)\\
     &\quad + P(|y_{i} - x_i^\top \beta_{\ast}|>\tilde{\epsilon}_m^{-1/2}) + P(\Omega_m^{c}). 
 \end{align*}
 From Lemmas \ref{lemma: psudo-psi-conv} and \ref{lemma: latent-bias-bound}, we have
 \begin{align*}
     P(\Omega_{m}^{c}) &\leq P(\|\psih^{(\gamma)} -\psi_{\ast}\|>\tilde{\epsilon}_m)\\
     &\leq P(\|\psih^{(\gamma)} -\psi_{\dagger}^{(\gamma)}\|>\tilde{\epsilon}_m/2) + P(\|\psi_{\dagger}^{(\gamma)} -\psi_{\ast}\|>\tilde{\epsilon}_m/2)\\
     &\leq 4\tilde{\epsilon}_m^{-2}\left(E[\|\psih^{(\gamma)} -\psi_{\dagger}^{(\gamma)}\|^{2}] + \|\psi_{\dagger}^{(\gamma)} -\psi_{\ast}\|^{2}\right)\\
     &= O\left(\frac{1}{m(m^{-\zeta/2} + \omega\rho_{\psi_{\ast}})} + \frac{(\omega\rho_{\psi_{\ast}})^{2}}{m^{-\zeta/2} + \omega\rho_{\psi_{\ast}}}\right)\\
     &= O\left(\frac{1}{m\omega\rho_{\psi_{\ast}}} + \omega\rho_{\psi_{\ast}}\right) = O(\omega \rho_{\psi_{\ast}}). 
 \end{align*}
 Further, we have
 \begin{align*}
     P(|y_{i} - x_i^\top \beta_{\ast}|>\tilde{\epsilon}_m^{-1/2}) &= P\left(\left. \left\{|y_{i} - x_i^\top \beta_{\ast}|>\tilde{\epsilon}_m^{-1/2}\right\}\right|B_i=0\right)P(B_i=0)\\
     &\quad + P\left(\left. \left\{|y_{i} - x_i^\top \beta_{\ast}|>\tilde{\epsilon}_m^{-1/2}\right\} \right|B_i=1\right)P(B_i=1)\\
     &\leq P\left(|G_{i}- x_i^\top\beta_{\ast} + \sqrt{D_i}Z_i|>\tilde{\epsilon}_m^{-1/2}\right)(1-\omega) + \tilde{\epsilon}_{m}E_{\delta}[|y_i-x_i^\top\beta_{\ast}|^{2}]\omega\\
     &\leq \tilde{\epsilon}_m^{1/2}\left(E\left[|G_i-x_i^\top \beta_{\ast}|\right] + \sqrt{D_{U}}E\left[|Z_i|\right]\right)(1-\omega) + O(\omega\tilde{\epsilon}_m)\\ 
     &= O(\tilde{\epsilon}_{m}^{1/2}).
 \end{align*}
 Then we have
 \begin{align}
     &P\left(\th_i \leq \thh_{i}^{(\gamma)} + \sh_i^{(\gamma)} z_{\alpha/2}\right) \nonumber \\ 
     &= P\left(\left\{\th_i \leq \tht_i^{(\gamma)} + \st_i^{(\gamma)} z_{\alpha/2} + (\thh_{i}^{(\gamma)} - \tht_i^{(\gamma)}) + (\sh_i^{(\gamma)}-\st_i^{(\gamma)}) z_{\alpha/2}\right\}\cap \Omega_{m} \cap \{|y_{i} - x_i^\top\beta_{\ast}| \leq \tilde{\epsilon}_{m}^{-1/2}\}\right) \nonumber \\
     &\quad + O(\tilde{\epsilon}_{m}^{1/2}) + O(\omega \rho_{\psi_{\ast}}). \label{thh-bonund-coverage-2}
 \end{align}
 On $\Omega_{m} \cap \{|y_{i} - x_i^\top\beta_{\ast}| \leq \tilde{\epsilon}_{m}^{-1/2}\}$, we can show 
 \begin{align*}
     |\thh_{i}^{(\gamma)} - \tht_i^{(\gamma)}| &\leq C'_{1}(D_L, D_U,\psi_{\ast})\tilde{\epsilon}_{m}^{1/2},\\ |\sh_i^{(\gamma)}-\st_i^{(\gamma)}| &\leq C'_{1}(D_L, D_U,\psi_{\ast})\tilde{\epsilon}_{m}^{1/2}
 \end{align*}
 for $\gamma \in [\underline{\gamma},\overline{\gamma}]$ and some positive constant $C'_{1}(D_L, D_U,\psi_{\ast})$. Then we have
 \begin{align}
    &P\left(\left\{\th_i \leq \tht_i^{(\gamma)} + \st_i^{(\gamma)} z_{\alpha/2} + (\thh_{i}^{(\gamma)} - \tht_i^{(\gamma)}) + (\sh_i^{(\gamma)}-\st_i^{(\gamma)}) z_{\alpha/2}\right\}\cap \Omega_{m} \cap \{|y_{i} - x_i^\top\beta_{\ast}| \leq \tilde{\epsilon}_{m}^{-1/2}\}\right) \nonumber \\
    &\leq P\left(\th_i \leq \tht_i^{(\gamma)} + \st_i^{(\gamma)} z_{\alpha/2} + C'_{1}(D_L, D_U,\psi_{\ast})\tilde{\epsilon}_{m}^{1/2}\right)\nonumber \\
    &=P\left(\left. \th_i \leq \tht_i^{(\gamma)} + \st_i^{(\gamma)} z_{\alpha/2} + C'_{1}(D_L, D_U,\psi_{\ast})\tilde{\epsilon}_{m}^{1/2}\right|B_i=0\right)P(B_i=0) \nonumber \\
    &\quad + P\left(\left. \th_i \leq \tht_i^{(\gamma)} + \st_i^{(\gamma)} z_{\alpha/2} + C'_{1}(D_L, D_U,\psi_{\ast})\tilde{\epsilon}_{m}^{1/2}\right|B_i=1\right)P(B_i=1) \nonumber \\
    &=: P_{1}(1-\omega) + P_{2}\omega \label{P1-P2-decomp}
 \end{align}
 For $P_1$, we have
 \begin{align}
     P_1 &= P\left(\left. \th_i \leq \tht_i + \st_i z_{\alpha/2} + (\tht_i^{(\gamma)}-\tht_i) + (\st_i^{(\gamma)} - \st_i) z_{\alpha/2} + C'_{1}(D_L, D_U,\psi_{\ast})\tilde{\epsilon}_{m}^{1/2}\right|B_i=0\right) \nonumber \\
     &= E\left[P( \th_i \leq \tht_i + \st_i z_{\alpha/2} + (\tht_i^{(\gamma)}-\tht_i) + (\st_i^{(\gamma)} - \st_i) z_{\alpha/2} \right. \nonumber \\ 
     &\left. \quad \quad \quad + C'_{1}(D_L, D_U,\psi_{\ast})\tilde{\epsilon}_{m}^{1/2}|y_i,x_i,D_i,B_i=0)|B_i=0\right] \nonumber \\
     &\leq E\left[P( \th_i \leq \tht_i + \st_i z_{\alpha/2}  |y_i,x_i,D_i,B_i=0) \right. \nonumber \\
     &\left. \quad + 2|\tht_i^{(\gamma)}-\tht_i| + |\st_i^{(\gamma)} - \st_i| z_{\alpha/2}+2C'_{1}(D_L, D_U,\psi_{\ast})\tilde{\epsilon}_{m}^{1/2}|B_i=0\right] \nonumber \\
     &= 1-\alpha/2 + 2E_{\psi_{\ast}}\left[|\tht_i^{(\gamma)}-\tht_i| + |\st_i^{(\gamma)} - \st_i| z_{\alpha/2}\right] + 2C'_{1}(D_L, D_U,\psi_{\ast})\tilde{\epsilon}_{m}^{1/2} \label{P1-bound}
 \end{align}
 by Lemma \ref{Nazarov-ineq}.
 Note that 
 \begin{align*}
     \tht_i^{(\gamma)} - \tht_i &= \frac{D_i}{A_{\ast} + D_i}\left(\phi(y_{i};x_i^\top \beta_{\ast},A_{\ast}+D_i)^{\gamma}c_{\gamma}(A_{\ast})-1\right)(y_i - x_i^\top \beta_{\ast}),\\
     \st_i^{2(\gamma)} - \st_i^{2} &= \frac{D_i^2}{A_{\ast} + D_i}\left\{\left(1- \phi(y_{i};x_i^\top \beta_{\ast},A_{\ast}+D_i)^{\gamma}c_{\gamma}(A_{\ast})\right)\right. \\
     &\left. \quad \quad + \frac{\gamma(y_i-x_i^\top \beta_{\ast})^{2}}{A_{\ast}+D_i}\phi(y_{i};x_i^\top \beta_{\ast},A_{\ast}+D_i)^{\gamma}c_{\gamma}(A_{\ast})\right\},
 \end{align*}
 and 
 \begin{align*}
     |\st_i^{(\gamma)} - \st_i| &\leq \st_i\left|\frac{\st_i^{(\gamma)}}{\st_i}-1\right|\left|\frac{\st_i^{(\gamma)}}{\st_i}+1\right| = \frac{1}{\st_i}\left|\st_i^{2(\gamma)} - \st_i^{2}\right| \leq \frac{A_{\ast}+D_U}{D_L}\left|\st_i^{2(\gamma)} - \st_i^{2}\right|.
 \end{align*}
 Since 
 \begin{align*}
     \frac{\partial }{\partial \gamma}\phi(y_{i};x_i^\top \beta_{\ast},A_{\ast}+D_i)^{\gamma} &= -\frac{1}{2}\phi(y_{i};x_i^\top \beta_{\ast},A_{\ast}+D_i)^{\gamma}\left(\log(2\pi(A_{\ast}+D_i)) + \frac{(y_i - x_i^\top \beta_{\ast})^{2}}{A_{\ast}+D_i}\right),\\
     \frac{\partial }{\partial \gamma}c_{\gamma}(A_{\ast}) &= c_{\gamma}(A_{\ast})\frac{\gamma(\gamma+2)}{2(1+\gamma)^{2}}\log(2\pi(A_{\ast}+D_i)),
 \end{align*}
 we have
 \begin{align}
     |\tht_i^{(\gamma)} - \tht_i| &\leq K_{1}(D_L, D_U,\psi_{\ast})\left(1 + \frac{(y_i - x_i^\top \beta_{\ast})^{2}}{A_{\ast}+D_i}\right)(1+|y_i - x_i^\top \beta_{\ast}|)\gamma,  \label{th-g-error} \\
     |\st_i^{2(\gamma)} - \st_i^{2}| &\leq K_{1}(D_L, D_U,\psi_{\ast})\left(1 + \frac{(y_i - x_i^\top \beta_{\ast})^{2}}{A_{\ast}+D_i}\right)(1+|y_i - x_i^\top \beta_{\ast}|)\gamma \label{s-g-error}
 \end{align}
 for $\gamma \in [\underline{\gamma},\overline{\gamma}]$ and some positive constant $K_{1}(D_L, D_U, \psi_{\ast})$. Combining (\ref{P1-bound}), (\ref{th-g-error}), and (\ref{s-g-error}), we have 
 \begin{align}\label{P1-bound-2}
     P_1 &\leq 1-\alpha/2 + O(\gamma) + 2C'_{1}(D_L, D_U,\psi_{\ast})\tilde{\epsilon}_{m}^{1/2}.
 \end{align}
 Define $\Omega_{2,m} = \{\phi(y_{i};x_i^\top \beta_{\ast},A_{\ast}+D_i)^{\gamma}|y_i-x_i^\top \beta_{\ast}|^{k}\leq \rho_{\psi_{\ast}}^{1/2}, k=0,1,2\}$. On $\Omega_{2,m}$, we have
 \begin{align*}
     \tht_i^{(\gamma)} &\leq  \theta_i + \sqrt{D_i} Z_i + \frac{D_i}{A_{\ast}}C_{\gamma}(A_{\ast})\rho_{\psi_{\ast}}^{1/2},\ \st_i^{2(\gamma)} \leq D_i\left(1 + \frac{D_i}{A'_{\ast}}(\gamma + 1)C_{\gamma}(A_{\ast})\rho_{\psi_{\ast}}^{1/2}\right)
 \end{align*}
 where $A'_{\ast} = \min\{A_\ast, A_\ast^2\}$. Then we have 
 \begin{align*}
     P_{2} &= P\left(\left\{\left. \th_i \leq \tht_i^{(\gamma)} + \st_i^{(\gamma)} z_{\alpha/2} + C'_{1}(D_L, D_U,\psi_{\ast})\tilde{\epsilon}_{m}^{1/2}\right\}\cap \Omega_{2,m}\right|B_i=1\right)\\
     &\quad + P\left(\left\{\left. \th_i \leq \tht_i^{(\gamma)} + \st_i^{(\gamma)} z_{\alpha/2} + C'_{1}(D_L, D_U,\psi_{\ast})\tilde{\epsilon}_{m}^{1/2}\right\}\cap \Omega_{2,m}^{c}\right|B_i=1\right)\\
     &=: P_{2,1} + P_{2,2}.
 \end{align*}
 For $P_{2,1}$, we have
 \begin{align}
     &P_{2,1} \nonumber \\
     &\leq P\left( \theta_i \leq \theta_i + \sqrt{D_i}Z_i + \frac{D_i}{A_{\ast}}C_{\gamma}(A_{\ast})\rho_{\psi_{\ast}}^{1/2} + \sqrt{D_i}\left(1 + \frac{D_i}{A'_{\ast}}(\gamma + 1)C_{\gamma}(A_{\ast})\rho_{\psi_{\ast}}^{1/2}\right)^{1/2}z_{\alpha/2} \right. \nonumber \\
     &\left. \quad \quad + C'_{1}(D_L, D_U,\psi_{\ast})\tilde{\epsilon}_{m}^{1/2}| B_i=1 \right) \nonumber\\
     &= P\left( -\sqrt{D_i} Z_{i} \leq \frac{D_i}{A_{\ast}}C_{\gamma}(A_{\ast})\rho_{\psi_{\ast}}^{1/2} + \sqrt{D_i}\left(1 + \frac{D_i}{A'_{\ast}}(\gamma + 1)C_{\gamma}(A_{\ast})\rho_{\psi_{\ast}}^{1/2}\right)^{1/2}z_{\alpha/2}\right. \nonumber \\
     &\left. \quad \quad + C'_{1}(D_L, D_U,\psi_{\ast})\tilde{\epsilon}_{m}^{1/2}| B_i=1 \right) \nonumber\\ 
     &= P\left(-Z_{i} \leq \frac{\sqrt{D_i}}{A_{\ast}}C_{\gamma}(A_{\ast})\rho_{\psi_{\ast}}^{1/2} + \left(1 + \frac{D_i}{A'_{\ast}}(\gamma + 1)C_{\gamma}(A_{\ast})\rho_{\psi_{\ast}}^{1/2}\right)^{1/2}z_{\alpha/2} \right. \nonumber \\
     &\left. \quad \quad + C'_{1}(D_L, D_U,\psi_{\ast})\tilde{\epsilon}_{m}^{1/2}| B_i=1 \right) \nonumber\\
     &\leq P\left(\!-Z_{i} \leq \!\left(\!1 + \frac{D_U}{A'_{\ast}}(\gamma + 1)C_{\gamma}(A_{\ast})\rho_{\psi_{\ast}}^{1/2}\right)^{1/2}\!\!\!\!\!\!\! z_{\alpha/2}\!\right)\! + \frac{2\sqrt{D_{U}}}{A_{\ast}}C_{\gamma}(A_{\ast})\rho_{\psi_{\ast}}^{1/2} + 2C'_{1}(D_L, D_U,\psi_{\ast})\tilde{\epsilon}_{m}^{1/2} \nonumber \\
     &\leq P\left(-Z_{i} \leq \left(1 + \frac{D_U}{A'_{\ast}}(\gamma + 1)C_{\gamma}(A_{\ast})\rho_{\psi_{\ast}}^{1/2}\right)z_{\alpha/2}\right) + \frac{2\sqrt{D_{U}}}{A_{\ast}}C_{\gamma}(A_{\ast})\rho_{\psi_{\ast}}^{1/2} + 2C'_{1}(D_L, D_U,\psi_{\ast})\tilde{\epsilon}_{m}^{1/2} \nonumber \\  
     &\leq P\left(-Z_{i} \leq z_{\alpha/2}\right) + \frac{2\sqrt{D_U}}{A'_{\ast}}C_{\gamma}(A_{\ast})\left((\gamma + 1)z_{\alpha/2} +1\right)\rho_{\psi_{\ast}}^{1/2} + 2C'_{1}(D_L, D_U,\psi_{\ast})\tilde{\epsilon}_{m}^{1/2} \nonumber  \\
     &\leq 1-\alpha/2 + O(\rho_{\psi_{\ast}}^{1/2} + \tilde{\epsilon}_{m}^{1/2}). \label{P21-bound}
 \end{align}
 For the second and fourth inequalities, we used Lemma \ref{Nazarov-ineq}. 
 
 For $P_{2,2}$, from the definition of $\rho_{\psi_{\ast}}$, we have 
 \begin{align}
     P_{2,2} &\leq P(\Omega_{2,m}^{c}|B_i=1) \nonumber \\
     & = E[P(\Omega_{2,m}^{c}|x_i,D_i, B_i=1)|B_i=1] \nonumber \\
     &\leq \rho_{\psi_{\ast}}^{-1/2}\max_{k \in \{0,1,2\}}\sup_{D_i \in [D_L, D_U]}\sup_{x_i \in \mathcal{X}}E_{\delta}[\phi(y_{i};x_i^\top \beta_{\ast},A_{\ast}+D_i)^{\gamma}|y_i-x_i^\top \beta_{\ast}|^{k}|x_i,D_i] \nonumber \\
     &\leq \rho_{\psi_{\ast}}^{1/2}.\label{P22-bound}
 \end{align}
 Combining (\ref{P21-bound}) and (\ref{P22-bound}), we have
 \begin{align}\label{P2-bound}
     P_{2} &\leq 1-\alpha/2 + O(\rho_{\psi_{\ast}}^{1/2}+\tilde{\epsilon}_{m}^{1/2}).
 \end{align}
 Combining (\ref{thh-bonund-coverage-2}), (\ref{P1-P2-decomp}), (\ref{P1-bound-2}), and (\ref{P2-bound}), we have 
 \begin{align*}
     P\left(\th_i \leq \thh_{i}^{(\gamma)} + \sh_i^{(\gamma)} z_{\alpha/2}\right)&\leq 1-\alpha/2 + O(\gamma) + O((\omega \rho_{\psi_{\ast}})^{1/4}).
 \end{align*}
 Likewise, we have $P\left(\th_i \leq \thh_{i}^{(\gamma)} + \sh_i^{(\gamma)} z_{\alpha/2}\right) \geq 1-\alpha/2 - O\left((\omega \rho_{\psi_{\ast}})^{1/4}\right) - O(\gamma)$. 
 Almost the same argument yields $P\left(\th_i \leq \thh_{i} - \sh_i z_{\alpha/2}\right) = \alpha/2 + O\left((\omega \rho_{\psi_{\ast}})^{1/4}\right) + O(\gamma)$.
 Therefore, we have
 \begin{align*}
     P\left(\th_i \in \widehat{I}_{i,1-\alpha}^{(\gamma)}\right) = 1-\alpha + O\left((\omega \rho_{\psi_{\ast}})^{1/4} + \gamma \right). 
 \end{align*}

\section{Additional results}\label{sec:sim-add}

\subsection{Further discussion on coverage accuracy in Theorem \ref{contami-coverage}}

Under the same assumptions of Theorem \ref{contami-coverage}, it is possible to show 
\begin{align}
    P\left(\theta_i \in \widehat{I}_{i,1-\alpha}^{(\gamma)}\right) & \geq (1-\omega)P_{\psi_{\ast}}\left(r_{i,-}^{(\gamma)}  -z_{\alpha/2} \leq \frac{\th_i - \tht_i}{\st_i} \leq z_{\alpha/2} + r_{i,+}^{(\gamma)} \right) \nonumber \\
    &\quad + \omega(1-\alpha) + O((\omega \rho_{\psi_{\ast}})^{1/4}), \label{coverage-discuss}
\end{align}
where $P_{\psi_{\ast}}(\cdot)$ denotes the probability with respect to $y_i \sim N(x_i^\top \beta_{\ast}, A_{\ast}+D_i)$ and 
\begin{align*}
    r_{i,-}^{(\gamma)} &= \frac{(\tht_i^{(\gamma)} - \tht_i) - (\st_i^{(\gamma)}-\st_i)z_{\alpha/2}}{\st_i},\ r_{i,+}^{(\gamma)} = \frac{(\tht_i^{(\gamma)} - \tht_i) + (\st_i^{(\gamma)}-\st_i)z_{\alpha/2}}{\st_i}.
\end{align*}
Note that 
\begin{align*}
    &\frac{\partial}{\partial \gamma}\phi(y_i;x_i^\top\beta_{\ast},A_{\ast}+D_i)^{\gamma}c_{\gamma}(A_{\ast})\\ 
    &= -\frac{1}{2}\phi(y_i;x_i^\top\beta_{\ast},A_{\ast}+D_i)^{\gamma}c_{\gamma}(A_{\ast})\left(\frac{1}{(1+\gamma)^{2}}\log(2\pi(A_{\ast}+D_i)) + \frac{(y_i - x_i^\top\beta_{\ast})^{2}}{A_{\ast}+D_i}\right)
\end{align*}
and 
\begin{align*}
\st_i^{2(\gamma)} - \st_i^{2} 
&= \frac{D_i^2}{A_{\ast} + D_i}\bigg\{\left(1- \phi(y_{i};x_i^\top\beta_{\ast},A_{\ast}+D_i)^{\gamma}c_{\gamma}(A_{\ast})\right)\\
& \ \ \ \ 
+ \frac{\gamma(y_i-x_i^\top\beta_{\ast})^{2}}{A_{\ast}+D_i}\phi(y_{i};x_i^\top\beta_{\ast},A_{\ast}+D_i)^{\gamma}c_{\gamma}(A_{\ast})\bigg\}.
\end{align*}
Thus, if we also assume $A_{\ast} + D_L>(2\pi)^{-1}$, we have $\frac{\partial}{\partial \gamma}\phi(y_i;x_i^\top\beta_{\ast},A_{\ast}+D_i)^{\gamma}c_{\gamma}(A_{\ast})<0$ for $\gamma>0$ and this implies $\st_i^{(\gamma)} - \st_i >0$ for $\gamma>0$. Moreover, in the proof of Theorem \ref{contami-coverage}, the first term in the right hand side of (\ref{coverage-discuss}) can be evaluated as $(1-\omega)(1-\alpha + O(\gamma))$.

\subsection{Additional simulation results}
Using simulation studies with the same settings in Section~\ref{sec:sim}, we here demonstrate that the empirical coverage probabilities of the proposed confidence intervals tend to be larger than the nominal ones when the first model (\ref{model}) is correctly specified and $\gamma$ is positive. 
Based on 2000 Monte Carlo replications, we computed CP and AL of the proposed confidence interval with each fixed value of $\gamma\in \{0, 0.1, 0.2, 0.3\}$.
The results are presented in Table \ref{tab:sim-add}, which clearly shows that both CP and AL tends to increase with $\gamma$, that is, the interval gets more inefficient by using a larger value of $\gamma$. 
This phenomena is related to the trade-off between efficiency under correct specification and robustness under contamination of the $\gamma$-divergence with positive $\gamma$.

\begin{table}[!htbp]
\caption{Empirical coverage probability (CP) and average length (AL) of the $95\%$ confidence intervals produced by the proposed GD method with $\gamma\in\{0, 0.1, 0.2, 0.3\}$ under Scenario (i) in Section~\ref{sec:sim}.  
\label{tab:sim-add}
}
\begin{center}
\begin{tabular}{ccccccccccccccc}
\hline
& \multicolumn{4}{c}{CP} & & \multicolumn{4}{c}{AL}\\
$\gamma$ & $0$ & $0.1$ & $0.2$ & $0.3$ &  & $0$ & $0.1$ & $0.2$ & $0.3$ \\
 \hline
$A=1$ & 93.6 & 96.0 & 96.6 & 96.7 &  & 2.54 & 2.90 & 3.12 & 3.26 \\
$A=0.5$ & 92.0 & 96.4 & 97.2 & 97.3 &  & 2.03 & 2.56 & 2.85 & 3.05 \\
\hline
\end{tabular}
\end{center}
\end{table}

\bigskip
\noindent
{\bf {\Large References}}

\vspace{0.5cm}
\noindent
Chernozhukov, V., D. Chetverikov, and K. Kato (2017). 
Central limit theorems and bootstrap in high dimensions. {\it The Annals of Probability} 45, 2309-2352.

\vspace{0.5cm}
\noindent
Fujisawa,  H.  (2013).
Normalized  estimating  equation  for  robust  parameter  estimation.
{\it Electronic Journal of Statistics} 7, 1587-1606.

\vspace{0.5cm}
\noindent
Fujisawa, H. and S. Eguchi (2008). Robust parameter estimation with a small bias against heavy contamination. {\it Journal of Multivariate Analysis} 99(9), 2053-2081.

\vspace{0.5cm}
\noindent
Nazarov,  F. (2003). 
On the maximal perimeter of a convex set in $\mathbb{R}^{m}$ with respect to a gaussian measure.
{\it Geometric Aspects of Functional Analysis} 1807, 169-187.

\vspace{0.5cm}
\noindent
van der Vaart, A. W. (1998). {\it Asymptotic Statistics}, Cambridge University Press.


\end{document}